%% file: v2_main_draft.tex
\documentclass[letterpaper, 10 pt, conference]{ieeeconf}
\IEEEoverridecommandlockouts

\usepackage{graphicx}
\usepackage{multirow}
\usepackage[nolist,printonlyused]{acronym}
\usepackage{amsmath, amsfonts}
\usepackage{algorithmic}
\usepackage[ruled,vlined]{algorithm2e}
\usepackage{bbold}
\usepackage{bm}
\usepackage{booktabs}
\usepackage[english]{datetime2}
\usepackage{glossaries}
\usepackage{multicol}
\usepackage{nicefrac}
\usepackage{textcomp}
\usepackage{siunitx}
\usepackage[caption=false,font=footnotesize]{subfig}
\usepackage{cite}
\usepackage{cleveref}
\usepackage{mathtools}
\usepackage{xcolor}
\usepackage{url}
\usepackage[rightcaption]{sidecap}
\sidecaptionvpos{figure}{t}

\newtheorem{proposition}{Proposition}

\newtheorem{remark}{Remark}

\input{preamble}

\title{Level-2 Inverse Games for Inferring Agents' Estimates of Others' Objectives
}

\author{
Hamzah I. Khan${}^*$,
Jingqi Li${}^*$,
and David Fridovich-Keil%
\thanks{Hamzah I. Khan, Jingqi Li, and David Fridovich-Keil are with the Department of Aerospace Engineering and Engineering Mechanics, The University of Texas at Austin, Austin, TX 78712 USA 
(e-mail: hamzah@utexas.edu; jingqi.li@austin.utexas.edu; dfk@utexas.edu). ${}^*$ signifies equal contribution.}%
\thanks{This work was supported by the National Science Foundation under Grant No. 2336840 and by the Army Research Laboratory under Cooperative Agreement No. W911NF-23-2-0011.}
}

\begin{document}

\maketitle

\begin{abstract}
\input{include/abstract}
\end{abstract}

% \begin{IEEEkeywords}
% Multi-agent interactions, inverse games, dynamic games, incomplete information, theory of mind.
% \end{IEEEkeywords}

\input{include/sec1_intro}
\input{include/sec3_problem_formulation}
\input{include/sec4_approach}

\input{include/sec5_experiments_new}

% \newpage
\input{include/sec6_conclusion}
\bibliographystyle{ieeetr}
\footnotesize
\bibliography{references,lab_references}

% UNCOMMENT FOR POSTING TO ARXIV

% % \appendix
% \newpage
% % ~
% % \newpage
\section*{Supplementary Material}
\normalsize
\label{sec:supplementary-appendix}
In Supplement A, we include the required definitions and proofs for \Cref{prop:nonconvex main text} and \Cref{prop:bounds}. 
% \todo{A full table of all variables.}
% , as well as required definitions for them.
We include details of the experiments on the LQ game in Supplement B, and details of the lane change scenario in Supplement C.
Lastly, we include tables that organize and document every variable introduced in the main document in Supplement D.

\input{include/proof_details}

\input{include/sec7_appendix}

\section{Supplement D: Notation Reference}
In \cref{tab:notation_game_setup,tab:notation_costs_params,tab:notation_level2,tab:notation_lq,tab:notation_homogeneous,tab:notation_mcp,tab:notation_algorithm_example}, we include a complete list of the notation used in this work, organized by category, symbol, definition, and description.

% \david{check formatting, be consistent with abbreviations}
% \todo{Write this and reference earlier.}
% \section*{List of Symbols}
% \begin{description}
%   \item[$\alpha$] Angle of incidence
%   \item[$\beta$] Angle of refraction
%   \item[$n$] Refractive index
% \end{description}
\input{include/notation_tables}

\end{document}

%% file: preamble.tex
\crefname{figure}{Fig.}{Figs.}
\crefname{equation}{}{}

\usepackage{tabularx}

\crefname{table}{Table}{Tables}
\crefname{algorithm}{Alg.}{Algorithms}
\crefname{section}{Section}{Sections}
\crefname{remark}{Remark}{Remarks}
\newcommand\edit[1]{\textcolor{blue}{#1}}

\newcommand\cmt[1]{\textcolor{purple}{#1}}

\newcommand\david[1]{\textcolor{orange}{[DFK: #1]}}
\newcommand\jingqi[1]{\textcolor{blue}{[JL: #1]}}
\newcommand\yang[1]{\textcolor{pink}{[YT: #1]}}

\renewcommand\jingqi[1]{\textcolor{brown}{#1}}
\renewcommand\yang[1]{\textcolor{olive}{#1}}

\renewcommand\edit[1]{\textcolor{blue}{#1}}

\renewcommand\cmt[1]{}
\renewcommand\david[1]{_D_D}
\renewcommand\jingqi[1]{_J_J}
\renewcommand\yang[1]{_Y_Y}
\renewcommand\edit[1]{_E_E}

\newcommand{\T}{\top}

\newcommand\LQcontrolfunction{\bar{u}^{i,i}}
\newcommand\LQtotalvariable{\bar{\mathbf{z}}}
\newcommand\numLQtotalvariable{n_{\LQtotalvariable{}}}

\newcommand\secondLQtotalvariable{\bar{\mathbf{z}}}
\newcommand\totalvariable{\mathbf{\hat{z}}}
\newcommand\initstate[1]{x_{\textrm{init}}^{#1}}
\newcommand\xstate[2]{x^{#2}_{#1}}
\newcommand\StateTraj[1]{\mathbf{x}^{#1}}
\newcommand\obs[2]{y^{#2}_{#1}}
\newcommand\ObsTraj[1]{\mathbf{y}^{#1}}
\newcommand\ctrl[2]{u^{#2}_{#1}}
\newcommand\CtrlTraj[1]{\mathbf{u}^{#1}}

\newcommand\equalityconstraint[1]{h^{#1}}
\newcommand\inequalityconstraint[1]{g^{#1}}

\newcommand\ObservationFn{G}

\newcommand\dynamicsLagrangeVar{\lambda}
\newcommand\dynamicsInitLagrangeVar{\eta}
\newcommand\equalityConstraintLagrangeVar{\phi}
\newcommand\inequalityConstraintLagrangeVar{\mu}

\newcommand\dynamicsLagrange[2]{\dynamicsLagrangeVar{}^{#2}_{#1}}
\newcommand\dynamicsInitLagrange[1]{\dynamicsInitLagrangeVar{}^{#1}}

\newcommand\lagrangian[1]{\mathcal{L}^{#1}}
\newcommand\LQlagrangian[1]{\bar{\mathcal{L}}^{#1}}

\newcommand\DynamicsLagrange[1]{\boldsymbol{\dynamicsLagrangeVar{}}^{#1}}
\newcommand\DynamicsInitLagrange[1]{\dynamicsInitLagrangeVar{}^{#1}}

\newcommand\hypdynamicsLagrange[2]{\hat{\dynamicsLagrangeVar{}}^{#2}_{#1}}
\newcommand\hypdynamicsInitLagrange[1]{\hat{\dynamicsInitLagrangeVar{}}^{#1}}
\newcommand\hypequalityConstraintLagrange[1]{\boldsymbol{\hat{\equalityConstraintLagrangeVar{}}}^{#1}}
\newcommand\hypinequalityConstraintLagrange[1]{\boldsymbol{\hat{\inequalityConstraintLagrangeVar{}}}^{#1}}

\newcommand\hypDynamicsLagrange[1]{\boldsymbol{\hat{\dynamicsLagrangeVar{}}}^{#1}}
\newcommand\hypDynamicsInitLagrange[1]{\hat{\dynamicsInitLagrangeVar{}}^{#1}}
\newcommand\hypEqualityConstraintLagrange[1]{\boldsymbol{\hat{\equalityConstraintLagrangeVar{}}}^{#1}}
\newcommand\hypInequalityConstraintLagrange[1]{\boldsymbol{\hat{\inequalityConstraintLagrangeVar{}}}^{#1}}

\newcommand\numconstlagrange{n_g}
\newcommand\numequalitylagrange{n_h}
\newcommand\params[1]{\hat{\theta}^{#1}}
\newcommand\Params{\hat{\theta}}

\newcommand\gtparams[1]{\theta^{#1}}
\newcommand\GtParams{\theta}

\newcommand\secondgtparams[1]{\Theta^{#1}}

\newcommand\GtParamsguess[1]{\hat{\Theta}^{#1}}

\newcommand\LevelOneGtParamsguess{\bar{\Theta}}

\newcommand\loss{L}

\newcommand{\numplayers}{N}

\newcommand\paramspace[1]{\mathcal{P}^{#1}}
\newcommand\jointparamspace{\mathcal{P}}

\newcommand\numstates[1]{n_x^{#1}}
\newcommand\numctrls[1]{n_u^{#1}}
\newcommand\numparams[1]{n_p^{#1}}
\newcommand\numobs[1]{n_o^{#1}}

\newcommand\playeridx{i}
\newcommand\secondidx{j}

\newcommand\horizon{T}
\newcommand\dynamics[1]{f^{#1}}
\newcommand\costfn[1]{J^{#1}}
\newcommand\stagecostfn[2]{\ell_{#1}^{#2}}
\newcommand{\tsim}{\tau}

\newcommand{\Tsim}{\horizon{}_{\tsim{}}}

\newcommand\thetavec{\mathbb{\Theta}}

\newcommand\constructedlevelonesol{\check{\Theta}}

\newcommand{\initstateestimatenotime}[1]{\mathbf{\hat{\xstate{}{}}}_{0}}
\newcommand\lat{\text{lat}}
\newcommand\lon{\text{lon}}

\newcommand\hypxstate[2]{\hat{x}^{#2}_{#1}}
\newcommand\hypctrl[2]{\hat{u}^{#2}_{#1}}

\newcommand\hypothesizedStateTraj[1]{\mathbf{\hat{\xstate{}{}}}^{#1}}
\newcommand\hypothesizedCtrlTraj[1]{\mathbf{\hat{\ctrl{}{}}}^{#1}}

\newcommand\hypStateTraj[1]{\hypothesizedStateTraj{#1}}
\newcommand\hypCtrlTraj[1]{\hypothesizedCtrlTraj{#1}}

\newcommand\LinearKKTA{A}
\newcommand\LinearKKTB{B}

\newcommand\LinearKKTBarA{\bar{\LinearKKTA{}}}
\newcommand\LinearKKTBarB{\bar{\LinearKKTB{}}}
\newcommand\LinearKKTBarC{\bar{C}}

\newcommand\LinearKKTBlkDiagA{A}
\newcommand\LinearKKTBlkDiagB{B}
\newcommand\LinearKKTBlkDiagC{C}

\newcommand\LinearKKTBigBlkDiagA{\widehat{\LinearKKTBlkDiagA{}}}
\newcommand\LinearKKTBigBlkDiagB{\widehat{\LinearKKTBlkDiagB{}}}
\newcommand\LinearKKTBigBlkDiagC{\widehat{\LinearKKTBlkDiagC{}}}

%% file: include/abstract.tex
% \todo{Mention how dynamic games allow us to paramterize longer strategic decision making in continuous spaces with discrete parameters.}
% \todo{Update arxiv and citation bibtex for supplement before submitting.}
Effectively interpreting strategic interactions among multiple agents requires us to infer each agent’s objective from limited information.
% \todo{add sentence that explains the context of dynamic games? - repetitive process as in self-driving}
Existing inverse game-theoretic approaches frame this challenge in terms of a ``level-1’’ inference problem, in which we take the perspective of a third-party observer and assume individual agents share complete knowledge of one another’s objectives.
However, this assumption breaks down in decentralized, real-world scenarios like urban driving and bargaining, in which agents may act based on conflicting views of one another’s objectives. 
We demonstrate the necessity of inferring agents’ 
% \hmzh{what does this mean here?}
different estimates of each other’s objectives through empirical examples, and by theoretically characterizing the prediction error of level-1 inference on fictitious gameplay data from linear-quadratic games.
To address this fundamental issue, we propose a framework for level-\underline{2} inference to address the question: ``What does each agent believe about other agents’ objectives?’’
We prove that the level-2 inference problem is non-convex even in benign settings like linear-quadratic games, and we develop an efficient gradient-based approach for identifying local solutions. 
Experiments on a synthetic urban driving example show that our approach uncovers nuanced misalignments that level-1 methods miss.
% \vspace{-8pt}
% \todo{Move appendix stuff to main draft}

%% file: include/sec1_intro.tex
\section{Introduction}
\label{sec:intro}
% Autonomous agents operating in interactive environments must
% \david{$\leftarrow$ rm}
% In order to model the behavior of other decision-makers in interactive environments, 
Autonomous agents in interactive settings must develop an understanding of other decision-makers' objectives in order to model their behavior.
% \cite{pynadath2005psychsim}.
In this work, we explore the question
\begin{quote}\emph{
%From a third-party observer's perspective, 
What is each agent's objective \textit{and} what does every agent think are other agents' objectives?}
\end{quote}
% , which necessitates anticipating and understanding how agents make decisions given their underlying preferences and limitations.
%\cite{harbers2009modeling} \david{no reference is needed for such a generic statement. any reference for this is completely vacuous}% -- human or robotic -- to plan safely and comfortably \david{sounds too much like driving. that is a fine example, but i wouldn't open the paper by focusing on that in this venue} \todo{cite}.
% A critical requirement for this anticipation  \david{anticipating what?} is
%We motivate this statement \david{can easily tighten: ``For example...''} using an example from human driving: 
% For example, human drivers on the road do not have exact models for the behavior of every other vehicle but must instead infer them in real-time as they interact \cite{nori2024prevention}.
In urban driving, for example, failure to understand the intentions of another driver can easily lead to unsafe and/or inefficient behavior.
% Failures in properly inferring the possible motivations of other agents can result in crashes or other undesirable behavior.
% Consider, for example, the scenario depicted in \cref{fig:front-figure}(bottom).
% Consider the lane change scenario of 
Consider \cref{fig:front-figure}, in which the blue car wants to change into the lane currently occupied by the red car.
% we observe a red vehicle lane keeping and a blue vehicle changing from its initial lane into the top lane.
Initially, the two vehicles deadlock because both agents incorrectly estimate the other agent's target lane.
% that the other wants to switch lanes.
A third-party observer 
% (like the level-1 observer in \cref{fig:front-figure}) 
seeking to understand the interaction might conclude that both agents wish to stay in their lanes, cf. the ``level-1 observer'' in \cref{fig:front-figure}. %\david{cf. the ``level-1 observer'' in \cref{fig:front-figure}.}
% \hmzh{Is this the best motivation given our results in \cref{fig:lc-inference}?}
% \david{I like this example because I think it can illustrate clearly why the L2 model leads to better predictions than the incorrect L1 model.}
% \todo{Measure prediction cost later?} \hmzh{Won't this make reviewers ask for prediction costs (next sentence)?}
% In a more severe mismatch, if both agents seek to switch lanes but estimates the other's objective to stay in the same lane, it could lead to a crash.
However, that incorrect conclusion would lead to substantial error in predicting the agents' future behavior.
% , \hmzh{e.g., if slow-moving traffic were to appear and force the red car to slow down.
The level-2 observer (\cref{fig:front-figure}), by contrast, infers each agent's objective and their estimate of the other agent's objective, concluding that
% leading to a more coherent understanding:
% explanation of the depicted behavior:
both agents incorrectly estimate the other's target lane, and so they drive overcautiously and deadlock.
% the other wants to 
% that 
% \hmzh{where} 
% the other wants to switch lanes, and so they drive overcautiously and deadlock.
% \todo{does not follow from previous sentence}
% Thus, misunderstandings between agents in these complex interactions can significantly influence outcomes. \david{Suggest rephrasing to: ``
Fundamentally, the level-2 interaction model can explain observed behavior in terms of agents' \textit{misunderstanding} rather than only in terms of \textit{misaligned objectives}.

\input{figs/front_figure}

% Thus, these inference methods must consider what agents believe about other agents in order to fully understand interactions.
% Third-party observers like transportation engineers often monitor intersections with cameras \citep{morris2011lane} in order to improve road safety (i.e. if regular mismatches caused by confusing intersection design \cite{shirazi2016looking} lead to crashes).
% The fundamental aim of this work
% an autonomous vehicle or
This work aims to develop a formal mathematical model which can explain these interactions, and an algorithm that a third-party observer---like a city regulator seeking to optimize traffic patterns---can use to identify agents' estimates of one another's objectives.
To this end, we take a game-theoretic perspective and model the behavior of each agent as a rational best response to what it believes others will do.
Given such an interaction model, a third party with access to (possibly noisy, partial) observations of the agents' actions can solve an ``inverse game'' problem to identify parameters of each agent's objective function that best explain the data.
However, existing methods \cite{awasthi2020inverse,le2021lucidgames,peters2023learning,mehr2023maximum,li2023cost,liu2024auto} assume that, although the agents' objectives are unknown to the observer, they \emph{are} fully known to the interacting agents.
% themselves during the interaction.\yang{I am a little bit confused. Do you mean "observed agents know each other’s objectives?"} 
This ``level-1'' assumption is often unrealistic and ignores more nuanced explanations for observed behavior, as shown in \cref{fig:front-figure}. %readily breaks down in practice, 

Instead of assuming that all observed agents share knowledge of their objectives, our work extends the classical level-1 reasoning framework.
We incorporate the level-2 reasoning concept from the theory of mind literature \cite{gmytrasiewicz1995rigorous,harbers2009modeling}.
This allows us to model each agent as acting rationally based on its own objectives and its estimates of other agents' objectives.
Furthermore, we extend existing level-2 inference approaches 
% jin2024mmtom
\cite{karimi2020receding,fotiadis2021recursive,tian2021learning,tian2021anytime,levin2022bridging,ma2022recursive,he2025latent,jin2024mmtom,shi2025muma,jara2019theory} by
% Specifically, we 
generalizing them from discrete action and belief spaces to continuous spaces using a differentiable Mixed Complementarity Problem (MCP) formulation \cite{dirkse1995mcplib}.
% \emph{Specifically, we make three contributions.}
Specifically, we make \emph{three contributions}.
\begin{enumerate}
    \item We formalize a framework, based in theory of mind, for modeling level-2 inverse dynamic games using sets of coupled Nash equilibrium problems.
    We showcase its value for modeling scenarios in which agents have mismatched estimates of other agents' objectives.
    % We showcase the need for inferring agents' objectives and their estimates of other agents' objectives through empirical studies.
    \item We prove that the level-2 inference problem is non-convex and derive prediction error bounds for level-1 inference applied to linear-quadratic (LQ) game data generated by agents operating under a level-2 model.
    % theoretically motivating the need for level-2 inference over level-1 methods.%These results demonstrate that, unlike level-1 inference, level-2 inference effectively captures scenarios where agents have heterogeneous estimates of each other’s objectives.%We prove the non-convexity of the level-2 inference problem and derive prediction error bounds for level-1 inference on linear-quadratic (LQ) game data generated by agents using a level-2 model, which confirms that level-2 inference is capable of modeling agents with heterogeneous estimate of others' objective. %Our level-2 framework provides an overparameterization of level-1 inference, enabling it, by construction, to achieve lower prediction loss, and we theoretically explain this performance gap.%We prove the non-convexity of the level-2 inference problem and derive prediction error bounds for level-1 inference when applied to linear-quadratic (LQ) game data generated by agents operating under a level-2 model.
    
    % We analyze level-2 inference in games with linear dynamics and quadratic objectives, and we derive
    % % linear-quadratic (LQ) games, a special class of games with linear dynamics and quadratic costs, by deriving theoretical 
    % upper and lower bounds on the prediction error of level-1 inference when applied to data generated by agents operating under a level-2 model.
    % under heterogeneous estimates of others' objectives. 
    % Moreover, we prove the non-convexity of the level-2 inference problem.
    % that inferring the parameters of these level-2 models is, in general, a non-convex problem.
    \item We propose an efficient gradient-based algorithm for identifying locally optimal solutions to the level-2 inverse game problem.
    Our method outperforms level-1 inference in empirical studies on LQ games and a synthetic urban driving scenario by identifying agents' misaligned estimates of others' objectives. %Our empirical results on LQ games and a synthetic urban driving scenario demonstrate that our method outperforms level-1 inference, even under non-convex collision avoidance constraints.
    % We propose an efficient gradient-based algorithm for identifying locally-optimal solutions to the level-2 inverse game problem.
    % which allow for the inference of each player's estimates of other agents' objectives from behavior data. 
    % We demonstrate 
    % Our method outperforms level-1 inference on LQ games and a synthetic urban driving scenario by identifying agents' misaligned estimates of others' objectives, under non-convex constraints.
    % level-2 inference technique 
    % .
    % with nonlinear dynamics, 
    % showing that we can identify agents' misaligned estimates of others' objectives and .
    % Our framework to define the inverse level-2 game problem, 
    % \item and derive theoretical results connecting the equilibrium solutions of our model to those of the standard ``level-1'' model used in prior work. 
    % \item We propose an efficient, gradient-based algorithm that infers each agent’s objective and its estimates of other agents’ objectives that best explain the observed behavior.
\end{enumerate}

%% file: figs/front_figure.tex
\begin{figure}[t]
\centering
% \begin{minipage}[t]{0.48\textwidth}
% \vspace{0pt}
\centering
\includegraphics[width=\linewidth]{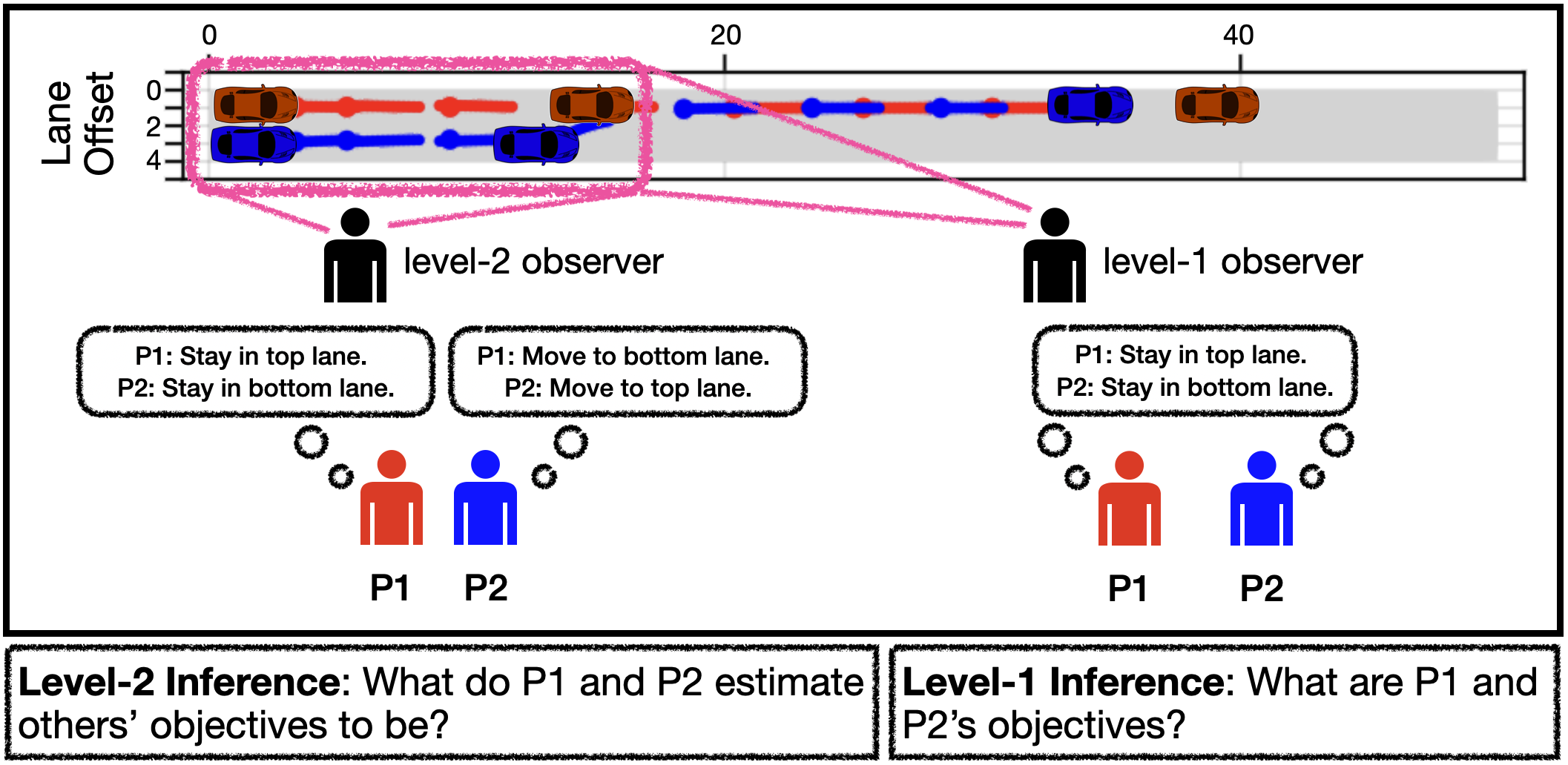}\vspace{-0.5em}
% \end{minipage}\hfill
% \begin{minipage}[t]{0.48\textwidth}
% \vspace{0pt}
\caption{\label{fig:front-figure}
\emph{Level-2 game inference.}
From observations of a multi-agent interaction, the level-2 observer infers each agent's objective and its estimate of others' objectives, capturing mismatched beliefs.
Level-1 methods assume agents know one another's objectives.
% \emph{Schematic of the proposed approach for level-2 game-theoretic model inference.}
% Given observations of a multi-agent interaction, our method infers each agent's objective \emph{and} its estimate of the others' objectives,
% thus accounting for potential mismatch in each agent's understanding of the interaction.
% By contrast, existing ``level-1'' methods assume that each agent knows the others' objectives and lack the flexibility to explain nuanced behavior that results from that mismatch.
}
\vspace{-1em}
% \end{minipage}
\end{figure}

%% file: include/sec3_problem_formulation.tex
\section{Preliminaries}
\label{sec:problem-formulation}

We consider a discrete-time dynamic game with $\numplayers$ agents interacting over a finite time horizon of length $\horizon$. For any positive integer $d$, define $[d] \equiv {1, 2, \ldots, d}$. Let $t \in [\horizon]$ index the discrete time steps. 

At each step $t$, agent $\playeridx$ has a state $\xstate{t}{\playeridx} \in \mathbb{R}^{\numstates{\playeridx}}$. The joint state vector is defined as $\xstate{t}{} \equiv \{ \xstate{t}{1}, \xstate{t}{2}, \ldots, \xstate{t}{\numplayers} \}$, whose total dimension is $\numstates{} = \sum_{\playeridx=1}^{\numplayers{}} \numstates{\playeridx}$.
We write the states of agent $\playeridx$ collected across the horizon $\horizon$ as $\StateTraj{\playeridx} \equiv \{\xstate{1}{\playeridx}, \xstate{2}{\playeridx}, \cdots, \xstate{\horizon}{\playeridx}\}$, and similarly define the game state trajectory (including all agents' states) over the horizon as $\StateTraj{} \equiv \{\xstate{1}{}, \xstate{2}{}, \cdots, \xstate{\horizon}{}\}$.
We denote $\xstate{t}{-i} \equiv \xstate{t}{} \setminus \{\xstate{t}{i}\}$ and $\StateTraj{-i} \equiv \StateTraj{} \setminus \{\StateTraj{i}\}$.
\emph{Throughout the work, we denote quantities which are aggregated over time with boldface (i.e. $\StateTraj{\playeridx}$).}
Each agent $i$'s state $\xstate{t}{i}$ evolves over time according to its dynamics $\xstate{t+1}{i} = \dynamics{i}(\xstate{t}{i}, \ctrl{t}{i})$,
% \begin{equation}
% \label{eq:dynamics}
%     \xstate{t+1}{i} = \dynamics{i}(\xstate{t}{i}, \ctrl{t}{i}),
% \end{equation}
where $\ctrl{t}{\playeridx} \in \mathbb{R}^{\numctrls{\playeridx}}$ denotes agent $\playeridx$'s actions at time $t$ and $\ctrl{t}{}$, $\CtrlTraj{\playeridx}$\!, and $\CtrlTraj{}$ are defined analogously to $\xstate{t}{}$, $\StateTraj{\playeridx}$\!, and $\StateTraj{}$.

\subsubsection{Nash Equilibria in Dynamic Games}
Let each agent $\playeridx$ have parameters $\gtparams{\playeridx} \in \paramspace{\playeridx} \subseteq \mathbb{R}^{\numparams{\playeridx}}$; we then denote the set of all agents' parameters by
\begin{equation}
\label{eq:param-space}
\GtParams{} \equiv \{ \gtparams{1}, \gtparams{2}, \cdots, \gtparams{\numplayers{}} \} \in \jointparamspace{} \equiv \paramspace{1} \times \paramspace{2} \times \cdots \times \paramspace{\numplayers{}}.
\end{equation}
Each agent $\playeridx{}$ minimizes a total objective $\costfn{\playeridx}(\StateTraj{}, \CtrlTraj{}; \gtparams{i}) \equiv \sum_{t=1}^T \stagecostfn{t}{i}(\xstate{t}{}, \ctrl{t}{}; \gtparams{\playeridx})$ that is defined in terms of stagewise objective functions $\stagecostfn{t}{i}(\xstate{t}{}, \ctrl{t}{}; \gtparams{\playeridx})$.
% , and the total cost is defined as $\costfn{\playeridx}(\StateTraj{}, \CtrlTraj{}; \gtparams{i}) \equiv \sum_{t=1}^T \stagecostfn{t}{i}(\xstate{t}{}, \ctrl{t}{}; \gtparams{\playeridx})$.
% \todo{player i's objective/stage cost only depends on its own set of parameters $\gtparams{i}/\gtparams{i,i}$. Double check everywhere to make it consistent.}
% \begin{equation}
%     \label{eq:objective}
%     \costfn{\playeridx}(\StateTraj{}, \CtrlTraj{}; \gtparams{i}) \equiv \sum_{t=1}^T \stagecostfn{t}{i}(\xstate{t}{}, \ctrl{t}{}; \gtparams{\playeridx})
% \end{equation}
This optimization is subject to the inequality constraint $\inequalityconstraint{\playeridx}(\StateTraj{}, \CtrlTraj{}) \ge 0$ and the equality constraint $\equalityconstraint{\playeridx{}}( \StateTraj{},\CtrlTraj{} ) = 0$ (which contains the agent's initial state constraint $\xstate{1}{i} = \initstate{i}$ and its dynamics constraints $\xstate{t+1}{i} = \dynamics{i}(\xstate{t}{i}, \ctrl{t}{i})$).
% \todo{Go through document and find instances of $\Gamma$.}
% , and dynamics $\dynamics{i}$.
% Given an initial state $\initstate{}$ and a horizon $T$, we denote the resulting $\numplayers{}$-agent parameterized Nash game as 
% \begin{equation}
% \label{eq:nash-game}
% \Gamma(\GtParams{}, \initstate{}, \horizon)
%     \equiv (
%     \numplayers{},  
%     \{\costfn{i}, 
%     \dynamics{i},\equalityconstraint{i}, 
%     \inequalityconstraint{i}\}_{i=1}^N 
%     ),
% \end{equation}
% and we simplify this notation to $\Gamma(\GtParams{})$ whenever the context is clear. 
We aim to optimize for $(\StateTraj{*}, \CtrlTraj{*})$ by jointly solving
\begin{equation}
    \label{eq:joint-optimization}
    \forall \playeridx \in  [\numplayers{}]\left\{
    \begin{aligned}
        \min_{\StateTraj{}, \CtrlTraj{\playeridx}} \ & \costfn{\playeridx}(\StateTraj{}, \CtrlTraj{}; \gtparams{\playeridx}) \\
        \text{s.t. } 0 &= h^{\playeridx}( \StateTraj{}, \CtrlTraj{}), 0 \leq g^{\playeridx}(\StateTraj{}, \CtrlTraj{}).
        % \forall t \in [\horizon-1]
        % , \xstate{1}{i} = \initstate{i}\\
        % & 
        % \xstate{t+1}{i} = \dynamics{i}(\xstate{t}{i}, \ctrl{t}{i}), \forall t \in [\horizon-1 ] .
    \end{aligned}
    % \vspace{-1em}
    \right.  
\end{equation}
If a solution $(\StateTraj{*}, \CtrlTraj{*})$ cannot be locally perturbed to unilaterally improve an agent's objective, i.e., the condition
\begin{equation*}
    \costfn{i}( \StateTraj{}, \CtrlTraj{\playeridx}, \CtrlTraj{(-\playeridx)*}; \gtparams{\playeridx} ) \geq \costfn{i}( \StateTraj{*}, \CtrlTraj{\playeridx*}, \CtrlTraj{(-\playeridx)*}; \gtparams{\playeridx} ), ~~ \forall i \in [N],
\end{equation*}
holds within a neighborhood around $\CtrlTraj{\playeridx*}$  (and the resulting state trajectory $\StateTraj{*}$) 
% \david{need the asterisks!!!!}
subject to the constraints in \cref{eq:joint-optimization},
% of $\Gamma(\GtParams{})$ \david{why not: ``constraints in (5)''?}, 
then we call $(\StateTraj{*}, \CtrlTraj{*})$ a \emph{local generalized Nash equilibrium} (LGNE) \cite{bacsar1998dynamic} 
% with components the LGNE state trajectory ($\StateTraj{*}$) and the LGNE 
 % strategy ($\CtrlTraj{*}$).
% . Moreover, we 
and refer to $\CtrlTraj{*}$ as the LGNE strategy and to $\StateTraj{*}$ as the LGNE state trajectory of the Nash game $\Gamma(\GtParams{})$.
% is referred to 
% In particular

\subsubsection{Inverse Dynamic Games}
\label{sssec:inverse-games}
We first introduce a well-known, related problem: inverse dynamic games \cite{molloy2017inverse,peters2021inferring,awasthi2020inverse}.
In an inverse dynamic game, an observer of the game (which may be an agent in it) receives (possibly noisy, partial, or missing) observations $\obs{t}{\playeridx} \in \mathbb{R}^{\numobs{\playeridx}} \cup \{\emptyset\}$ at each time $t$ for each agent $\playeridx$.
% \jingqi{Jingqi: we need to introduce $G$ the observation model.}
Analogous to states and actions in dynamic games, we write $\obs{t}{}$ for observations across all agents at time $t$, $\ObsTraj{\playeridx}$ for those of agent $\playeridx$ over time, and $\ObsTraj{}$ for all agents over all times. 
% \jingqi{We assume agents’ trajectory observations are generated under isotropic, zero-mean Gaussian noise, though this assumption can be relaxed to other noise models.}
% We denote collections of observations across all agents at a time, all times for an agent, and all times for all agents, i.e., $\obs{t}{}, \ObsTraj{\playeridx}, \ObsTraj{}$ as done for states and actions in dynamic games.
The aim of an inverse game is to estimate the true parameters $\GtParams{}$ of the game given these observations $\ObsTraj{}$. 
% \hmzh{We denote parameter estimates with hats, $\hat{\Box}$, e.g. $\Params{}$ is an estimate of $\GtParams{}$.
% Specifically, denoting the estimate} 
% \david{tighter: ``Denoting these estimates as''}
Denoting these estimates as
$\Params{} = \{ \params{1}, \params{2}, \cdots, \params{\numplayers{}} \}$, where $\params{\playeridx{}} \in \paramspace{\playeridx{}}$, the methods solve the maximum likelihood estimation (MLE) problem
\begin{subequations}
    \label{eq:level-1-inverse-game}
    \begin{align}
    \max_{\StateTraj{}, \CtrlTraj{}, \Params{}} ~ &
     p(\ObsTraj{} | \StateTraj{}, \CtrlTraj{}) \label{eq:level-1-inverse-game-objective} \\
     \text{s.t.} & ~~ (\StateTraj{}, \CtrlTraj{}) \in \text{LGNE of } \Gamma(\Params{}). \label{eq:level-1-inverse-nash-constraint}
     \end{align}
\end{subequations}
In this work, we assume the likelihood function $ p(\ObsTraj{} | \StateTraj{}, \CtrlTraj{}) $ is Gaussian with mean $ G(\StateTraj{}, \CtrlTraj{}) $ and isotropic diagonal covariance, so that its log-likelihood reduces to $\log (p(\ObsTraj{} | \StateTraj{}, \CtrlTraj{})) = -\|\ObsTraj{} - G(\StateTraj{}, \CtrlTraj{})\|_2^2.$
% \david{In this work, we assume the model $p(\ObsTraj{} | \StateTraj{}, \CtrlTraj{})$ is Gaussian with mean being an observation function $G(\StateTraj{}, \CtrlTraj{})$ and covariance be a diagonal matrix. Under this assumption, taking a negative log over the likelihood function, we have $-\log(p(\ObsTraj{} | \StateTraj{}, \CtrlTraj{}))$ can be rewritten as a quadratic objective $\|\ObsTraj{} - G(\StateTraj{}, \CtrlTraj{})\|_2$.}
% \jingqi{Jingqi: What about the partial observation term? Jingqi will edit this as per David's comment.}

Finally, please refer to Supplement D \cite{khan2025agentsthinkdolevel2} for a summary of all symbols introduced throughout the paper.

%% file: include/sec4_approach.tex
\newcommand\SubCtrlTraj[1]{\textbf{u}^{#1}}
\newcommand\subctrl[2]{\textbf{u}^{#2}_{#1}}

\newcommand{\iterindex}{r}
\newcommand{\obswindow}[1]{\mathbf{O}_{#1}}

\section{Formulating Level-2 Inverse Games}
\label{sec:approach}
% \todo{Do one read through where we confirm that usage of the words estimate and inference parameter are consistent with one another in the context they are used.}
% \subsection{Level-2 Inverse Games}
% In this section, we extend the classical (level-1) inverse game formulation to level-2 settings. Specifically, we are tasked to infer each agent $i$'s objective and also its estimate of others' objectives from historical behavior data, adopting the viewpoint of a third-party observer, a standard assumption in the inverse games literature. 
%\david{can trim prev sentence an combine here} 
% In this section, we extend the classical (level-1) inverse game formulation to the level-2 setting. Specifically, from the standard viewpoint of a third-party observer, we aim to infer each agent $i$'s objective as well as its estimates of other agents' objectives using historical behavior data $\obs{}{}$. 
% By doing so, we seek to answer the question:
% \begin{quote}\emph{
% %From a third-party observer's perspective, 
% What is each agent's objective \textit{and} what does every agent think are other agents' objectives? 
% %believe about % its %own and every other agent's 
% %given interaction observations?}
% }
% \end{quote}

% \todo{\textbf{Running Example}: Re-introduce the example from Fig. 1.}

% \todo{Clarify about the role of the third-party observer.}
% Before presenting the detailed formulation, w
We first define a level-2 game as one in which each agent $\playeridx$ knows its own true objective parameter $\gtparams{\playeridx, \playeridx}$ and maintains estimates of the other agents' objective parameters, denoted by  
% \david{in the experiments you don't have the comma in the superscript btw. i also feel like we don't really need to define $\theta^{i, -i}$ to have its own symbol} 
$\gtparams{i,-i}\equiv \{\gtparams{i,1}, \cdots, \gtparams{i,i-1}, \gtparams{i,i+1}, \cdots, \gtparams{i,N}\}$. Agent $i$'s parameters in a level-2 game are collectively represented by 
% $\secondgtparams{i} \equiv \{\gtparams{i}, \gtparams{i,-i}\}\in \jointparamspace{}$. 
\begin{equation}
    \secondgtparams{i} \equiv \{\gtparams{\playeridx, \playeridx}, \gtparams{i,-i}\}\in \jointparamspace{}.
\end{equation}
% \hmzh{Confusing notation below.}
Each agent $i$ independently computes a hypothesized LGNE of the game $\Gamma(\secondgtparams{i})$.
% parameterized by $\secondgtparams{i}$. 
We denote this hypothesized LGNE as $(\hypothesizedStateTraj{\playeridx}, \hypothesizedCtrlTraj{\playeridx})$, where
\begin{equation}
\label{eq:hypothesized-trajectories-for-i}
\hypothesizedStateTraj{\playeridx} \equiv\{\hypStateTraj{i,j}\}_{j=1}^{\numplayers{}} \text{ and } \hypothesizedCtrlTraj{\playeridx} \equiv \{ \hypCtrlTraj{i,j}  \}_{j=1}^{\numplayers{}}.
\end{equation} 
% \hmzh{though we may equivalently denote the LGNE as $(\hypothesizedStateTraj{}, \hypothesizedCtrlTraj{})$ for brevity when player $\playeridx$ is clear from context.}
Each player $\secondidx{}$'s \emph{hypothesized LGNE state trajectory and strategy under $\secondgtparams{i}$ are given by $\hypStateTraj{\playeridx{},\secondidx{}}$ and $\hypCtrlTraj{\playeridx{},\secondidx{}}$.}
% \david{this notation is not sufficiently clear. need to explain that $u^{i,j}$ is Pj's control sequence in Pi's hypothesized game} represent the \emph{hypothesized LGNE state trajectory and strategy} under $\secondgtparams{i}$. 
Subsequently, each agent $i$ selects its own actions $\hypCtrlTraj{i,i}\equiv\{\hypctrl{t}{i,i}\}_{t=1}^T$ from the hypothesized LGNE strategy $\hypothesizedCtrlTraj{\playeridx}$, and interacts with others accordingly. 

% \todo{\textbf{Running Example}: Describe what these are in the lane change and what game each player solves.}

In the corresponding inverse problem, the \emph{level-2 inverse game}, a third party observer infers the parameters $\secondgtparams{i}$ for each agent $i\in[N]$. To this end, for each agent $i$, we denote an estimate of $\secondgtparams{i}$ as
$\GtParamsguess{\playeridx{}} \equiv \{ \params{\playeridx{}, \playeridx{}}, \params{\playeridx{},-\playeridx{}}  \} \in \jointparamspace{}$.

%\david{didn't we already do this above?} \david{I think we need to rearrange so that (15) comes before the MCP bit, then we talk about MCP transcription, and we only talk about the parameters once}:
% \begin{equation}
%     \label{eq:level-2-sub-params}
%     \GtParamsguess{\playeridx{}} %= \params{\playeridx{}, [\numplayers{}]} 
%     \equiv \{ \params{\playeridx{}, \playeridx{}}, \params{\playeridx{},-\playeridx{}}  \} \in \jointparamspace{}.
% \end{equation}

% \begin{remark}
% While player $\playeridx{}$ is aware of his own parameter, $\gtparams{\playeridx{}, \playeridx{}}$, the third-party observer is not necessarily aware of it, necessitating the estimation of $\params{\playeridx{}, \playeridx{}}$ as in \cref{eq:level-2-sub-params}. \end{remark}

% To make the formulation tractable, inspired by the well-known fictitious-play perspective for game-theoretic decision-makings, 
% Combining one set of such parameters for each of the $\numplayers{}$ agents then specifies $\numplayers{}$ hypothesized Nash games, 
% $\Gamma(\GtParamsguess{1}), \Gamma(\GtParamsguess{2}), \cdots,  \Gamma(\GtParamsguess{N})$,
% one for each agent $\playeridx{} \in [\numplayers{}]$.
In an $N$-agent level-2 inverse game, the parameters to infer can be compactly represented as 
% \begin{equation}
% \label{eq:level-2-params}
$\GtParamsguess{} \equiv \left\{ \GtParamsguess{1}, \GtParamsguess{2}, \cdots, \GtParamsguess{N} \right\} \in \jointparamspace{} \times \jointparamspace{} \times\cdots \times \jointparamspace{}$,
% \end{equation}
and the goal is to estimate these parameters to maximize the likelihood of the observed trajectory $\ObsTraj{}$. The level-2 inverse game is thus formulated as:
\begin{subequations}
    \label{eq:level-2-inverse-game}
    \begin{align}
    \max_{\StateTraj{}, \CtrlTraj{}, \GtParamsguess{}} \ \  &
     p(\ObsTraj{} | \StateTraj{}, \CtrlTraj{} ) \label{eq:level-2-inverse-game-objective} \\
    \text{s.t. }   
    & ( \hypothesizedStateTraj{\playeridx}, \hypothesizedCtrlTraj{\playeridx} ) \in \text{LGNE of } \Gamma(\GtParamsguess{i}),  && \hspace{-6pt} \forall i \in [\numplayers{}], \label{eq:level-2-inverse-game-nash-constraint} \\
    %\xstate{1}{i}& =\initstate{i},
    & \CtrlTraj{} = \{ \hypCtrlTraj{1,1}, \hypCtrlTraj{2,2}, \cdots, \hypCtrlTraj{\numplayers{}, \numplayers{}} \}, \label{eq:level-2-inverse-game-ctrls-set} \\
    %& \xstate{t+1}{i}  = \dynamics{i}(\xstate{t}{i}, \subctrl{t}{i,i}), \forall t \in [\horizon{}-1], && \hspace{-6pt} \forall i\in[N], 
    & \StateTraj{} = \{ \hypStateTraj{1,1}, \hypStateTraj{2,2}, \cdots, \hypStateTraj{N,N} \}\label{eq:level-2-inverse-game-dynamics},
    \end{align}
\end{subequations}
% \jingqi{Jingqi: Should we hat the agent trajectories?}
% \todo{This is confusingly written.}
%\david{why not just replace the last constraint with a version of (10c) but for $x$. the current one is also imprecise bc the time indexing should be to [T-1] and there should be an initial state constraint. both these issues would be addressed with my suggested change]}
where the second constraint, \cref{eq:level-2-inverse-game-ctrls-set}, indicates that each agent~$i$ independently plans its actions~$\hypCtrlTraj{i,i}$ by extracting them from its hypothesized LGNE~$(\hypothesizedStateTraj{\playeridx}, \hypothesizedCtrlTraj{\playeridx})$. The joint action trajectory~$\CtrlTraj{}$ is then formed by aggregating these individual trajectories~$\{\hypctrl{t}{i,i}\}_{i=1}^N$. A similar interpretation holds for the constraint in \cref{eq:level-2-inverse-game-dynamics}. 

To analyze how the likelihood function varies with different hypothesized $\GtParamsguess{}$, we define a loss function purely in terms of $\GtParamsguess{}$, which we minimize to obtain the optimal level-2 inverse game solution:
 %$\loss( )$
\begin{equation}\label{eq:level-2 log-likelihood-inverse-objective}
\begin{aligned}
    \loss( \GtParamsguess{})\equiv -\max_{\StateTraj{}, \CtrlTraj{}}\  & p(\ObsTraj{} | \StateTraj{}, \CtrlTraj{}) \\
    \textrm{s.t. }& \eqref{eq:level-2-inverse-game-nash-constraint}, \eqref{eq:level-2-inverse-game-ctrls-set}, \eqref{eq:level-2-inverse-game-dynamics}
\end{aligned}
\end{equation}

We denote ground truth parameters for level-1 inference as $\gtparams{*}$ and for level-2 inference as $\secondgtparams{*}$.
% When there is no observation noise, the ground truth parameters minimize the inverse game loss.
% For convenience, we summarize the level-1 and level-2 parameters we have introduced in \cref{tab:parameters}.

% \todo{\textbf{Running Example}: Is anything more needed for the previous two paragraphs?}

% \input{figs/parameters_table}

% \todo{Move to a different section?}
\begin{remark}
\label{rmk:gradients}
We can locally improve our estimate $\Params^{\playeridx, \secondidx}$ using observations $\obs{}{}$ only if the gradient of the loss function $\cref{eq:level-2 log-likelihood-inverse-objective}$ with respect to $\Params^{\playeridx, \secondidx}$ is nonzero.
This gradient can become zero when observations are missing or if a small perturbation in $\Params^{\playeridx, \secondidx}$ does not change the likelihood of the observed trajectory $\ObsTraj{}$ (e.g., when an agent's equilibrium behavior is insensitive to its understanding of another agent's objective parameters).
% the Nash equilibrium does not depend on the parameter being inferred).}
% \hmzh{Inferring $\Params{\playeridx}$ using a gradient-based approach requires some level of coupling in the objectives or constraints of interacting agents.
% Let us separate the stage costs and constraints into equivalent individual and coupled terms, i.e.
% \begin{equation}
%     \stagecostfn{t}{\playeridx}(\xstate{t}{}, \ctrl{t}{}; \gtparams{\playeridx}) = \stagecostfn{t}{\playeridx{}, \text{indiv}}(\xstate{t}{\playeridx}, \ctrl{t}{\playeridx}; \gtparams{\playeridx}) + \stagecostfn{t}{\playeridx, \text{coupled}}(\xstate{t}{}, \ctrl{t}{}; \gtparams{\playeridx})
% \end{equation}
% \begin{equation}
%     \equalityconstraint{\playeridx}(\xstate{t}{}, \ctrl{t}{}) = \equalityconstraint{\playeridx{}, \text{indiv}}(\xstate{t}{\playeridx}, \ctrl{t}{\playeridx}) + \equalityconstraint{\playeridx, \text{coupled}}(\xstate{t}{}, \ctrl{t}{}) = 0
% \end{equation}
% \begin{equation}
%     \inequalityconstraint{\playeridx}(\xstate{t}{}, \ctrl{t}{}) = \inequalityconstraint{\playeridx{}, \text{indiv}}(\xstate{t}{\playeridx}, \ctrl{t}{\playeridx}) + \inequalityconstraint{\playeridx, \text{coupled}}(\xstate{t}{}, \ctrl{t}{}) \le 0
% \end{equation}
% }
\end{remark}

\section{Theoretical Characterization of Level-2 Inverse Games}
In this section, we focus on linear-quadratic (LQ) games, a class of dynamic games with linear dynamics and quadratic objectives that allow for analytical tractability. We show that level-2 inference is inherently non-convex, and
% . However, 
we demonstrate its necessity by deriving upper and lower bounds on the minimal prediction error of level-1 inference when the data involves agents with misaligned estimates of others’ objectives. % In this subsection, we theoretically characterize the limitations of existing level-1 inference methods when applied to datasets collected from agents with misaligned beliefs. We achieve this by deriving upper and lower bounds on the best achievable prediction error in LQ games.
All proofs are provided in the Appendix, and a reference for notation is provided in Supplement D
\cite{khan2025agentsthinkdolevel2}.
% \cite{khan2025agentsthinkdolevel2}.

We analyze level-2 inference in an $N$-agent finite-horizon linear–quadratic (LQ) dynamic game with quadratic costs and linear dynamics played over horizon $\horizon$.
Specifically, the $i$-th agent considers the quadratic stage cost
\begin{equation}
    \label{eq:lq-cost}
    \stagecostfn{t}{i}(\xstate{t}{}, \ctrl{t}{} ;\gtparams{i}) = \frac{1}{2} {\xstate{t}{}}^\top Q({\gtparams{i}}) \xstate{t}{} + \frac{1}{2} {\ctrl{t}{i}}^\top R^i \ctrl{t}{i}
\end{equation}
where the entries of the positive semidefinite matrix $Q({\gtparams{i}})$ are parameterized by the vector $\gtparams{i}$, and $R_t^i$ is a fixed positive definite matrix. 
%We define the total cost as $\costfn{\playeridx}(\StateTraj{}, \CtrlTraj{}; \gtparams{\playeridx})  = \sum_{t=1}^T \ell_t^i(x,u;\gtparams{i})$. 
Each agent $\playeridx$'s state dynamics are given by %, its dynamics $\dynamics{i}$ is
\begin{equation}
    \label{eq:lq-dyn}
    \xstate{t+1}{i} = \LinearKKTA{}^{\playeridx} \xstate{t}{i} + \LinearKKTB{}^{\playeridx} \ctrl{t}{i}. 
\end{equation}
To formulate the Lagrangian for each agent $i$, we associate the initial condition constraint $0=\xstate{1}{i}-\initstate{i}$ with a Lagrange multiplier $\dynamicsInitLagrange{i}$, and at each step $t$, we associate a Lagrange multiplier $ \dynamicsLagrange{t}{i}$ with dynamics constraint \cref{eq:lq-dyn}.
For clarity, we denote $\DynamicsInitLagrange{}\equiv\{\dynamicsInitLagrange{i}\}_{i=1}^N$, $\dynamicsLagrange{t}{}\equiv\{\dynamicsLagrange{t}{i}\}_{i=1}^N$, and $\DynamicsLagrange{}\equiv\{\dynamicsLagrange{t}{}\}_{t=1}^{T-1}$.
The Lagrangian specification, matrix definitions, and equilibrium conditions for the LQ game are provided in Appendix~A~\cite{khan2025agentsthinkdolevel2}.
Let $\LQtotalvariable{}$ encode the collection of all primal and dual variables, and $\numLQtotalvariable{}=T\cdot (\numstates{} + \numctrls{}) + (T-1)\cdot\numstates{} + \numstates{}$ be their total dimension. % of all the primal and dual variables.
% we considered in the KKT conditions. 
For convenient analysis, we assume throughout this section that there is a unique LGNE, meaning that $M(\GtParams{})$ is invertible. 
In practice, such an assumption is not required for the implementation of our method (as demonstrated by the lane change example in which multiple Nash equilibria can exist).
% \hmzh{Why do we assume these things?}
Precise conditions ensuring the uniqueness of the solution and invertibility of $M(\GtParams{})$ can be found in the text from \cite[\S6.2]{bacsar1998dynamic}.

To compute agent $i$'s LGNE control under its hypothesized game parameter $ \GtParamsguess{i} $, we have 
\begin{equation}\label{eq: LQ kkt equation}
    M( \GtParamsguess{i} ) \cdot \secondLQtotalvariable^i + S \cdot \initstate{} = 0
\end{equation}
where $\secondLQtotalvariable^i\in\mathbb{R}^{\numLQtotalvariable{}}$. Solving \eqref{eq: LQ kkt equation} yields $\secondLQtotalvariable^i(\GtParamsguess{i}) = -M(\GtParamsguess{i})^{-1} S \initstate{}$. Thus, agent $i$'s control $\LQcontrolfunction$, a subvector of $\secondLQtotalvariable^i$, can be explicitly expressed as the function $\LQcontrolfunction(\GtParamsguess{i})$.

% Similarly, we compute the LGNE solution $\LQtotalvariable'$ under level-1 inference as 
% \begin{equation}
%     M( \LevelOneGtParamsguess ) \LQtotalvariable' + S \initstate{} = 0
% \end{equation}
% In this subsection, we assume that 
Assuming that we have full observation of the entire action trajectory subject to zero-mean Gaussian noise, i.e., $\obs{t}{i} ~\sim \mathcal{N}(\ctrl{t}{i}, c^2 I)$ with variance $c^2 \in \mathbb{R}^+$, we can rewrite $\loss(\GtParamsguess{})$ as
\begin{equation}
    \label{eq:lq-inverse-game-loss}
    \loss(\GtParamsguess{}) =  \sum_{i=1}^N \sum_{t=1}^T \frac{1}{2} \| \LQcontrolfunction_t(\GtParamsguess{i}) - \obs{t}{i} \|_2^2.
\end{equation}
% We show in the following proposition that level-2 inference is a non-convex problem, even in this benign setting.
The following proposition shows level-2 inference is non-convex even in this benign setting.
% \hmzh{Please refer to 
% \david{Appendix A?}
% \cref{tab:notation_homogeneous,tab:notation_mcp}
% in the supplementary material \todo{\cite{khan2025agentsthinkdolevel2}} for a summary of all symbols used in this section.}
% \david{Reword as per earlier langauge.}\hmzh{For reader convenience, a reference for all notation is present in \cite{khan2025agentsthinkdolevel2}.}
\begin{proposition}\label{prop:nonconvex main text}
    Let $\{\secondgtparams{i*}\}_{i=1}^N$ be the ground truth parameters of a level-2 inverse game. Suppose that, in each agent $i$'s hypothesized game $\Gamma(\secondgtparams{i*})$, all agents' cost functions are strongly convex in $(\StateTraj{}, \CtrlTraj{})$. Define the observed data $\ObsTraj{}$ as the union of action trajectories $\hypCtrlTraj{i,i}$ extracted from the LGNE of each game $\Gamma(\secondgtparams{i*})$, as in \cref{eq:level-2-inverse-game-nash-constraint}. Then, $\loss(\GtParamsguess{})$ is a non-convex function.
    %Consider a set of ground truth parameters $\{\secondgtparams{i}\}_{i=1}^N$ for all agents in a level-2 game, suppose that the cost function $\costfn{i}(\StateTraj{}, \CtrlTraj{}; \gtparams{i})$ are strongly convex for $\StateTraj{}$ and $\CtrlTraj{}$, then the loss function of level-2 inference is nonconvex. 
\end{proposition}

% In what follows, we analyze the performance of level-1 inference when the collected observation data involves agents with misaligned estimate of others' objectives. To facilitate the analysis, we define the prediction loss under level-1 inference as 
% \todo{Add a table that explains what each parameter theta is.}
In what follows, we analyze the performance of level-1 inference when observed data includes agents with misaligned estimates of others' objectives. 
% \hmzh{the next sentence is confusing}
% \hmzh{Confusing notation - not clear that homogeneous means all the same $\theta^i$ and heterogeneous means can be different. Maybe add a table to explain?}
% Note that any level-1 inverse game parameter $\GtParams{}$ can be equivalently represented as a \emph{homogeneous} level-2 inverse game parameter.
We note that any level-1 inverse game parameter $\GtParams{}$ can be equivalently represented as a \emph{homogeneous} level-2 inverse game parameter $\LevelOneGtParamsguess$, {defined so that
% We define $\LevelOneGtParamsguess$ so that e
each player's parameters $\LevelOneGtParamsguess^\playeridx$ are identical to $\GtParams{}$, and %so 
the level-2 encoding of the level-1 problem is given by
% \hmzh{Hamzah: Fix the homogeneous definition. just have $\Theta^i = \theta \forall i$.}
% inverse game parameter representation of the level-1 inverse game parameters is given by}
\begin{equation}
\LevelOneGtParamsguess \equiv \{\LevelOneGtParamsguess^i\}_{i=1}^N; \quad \LevelOneGtParamsguess^\playeridx = \GtParams{} \quad 
% \!=\! \{\gtparams{\playeridx, l} \}_{l=1}^N}  \!=\! \LevelOneGtParamsguess^j \!=\! \GtParams{} \!=\! \{\theta^{l} \}_{l=1}^N \quad 
% \forall \playeridx, j \in [\numplayers{}].
\forall \playeridx \in [\numplayers{}].
\end{equation}
% \begin{equation}
% \LevelOneGtParamsguess^\playeridx = \LevelOneGtParamsguess^j = \GtParams{} \qquad \forall \playeridx, j \in [\numplayers{}]
% \end{equation}
% Thus, the level-2 inverse game parameter representation of the level-1 inverse game parameters is given by
% \begin{equation}
% \LevelOneGtParamsguess \equiv \{\LevelOneGtParamsguess^i\}_{i=1}^N; \quad \LevelOneGtParamsguess^i \equiv \{\gtparams{1}, \gtparams{2}, \dots, \gtparams{N}\} \forall \playeridx{} \in [\numplayers{}].
% \end{equation}
% \begin{equation}
% \LevelOneGtParamsguess \equiv \{\LevelOneGtParamsguess^\playeridx{} | \LevelOneGtParamsguess^\playeridx = \LevelOneGtParamsguess^j \forall j \in [\numplayers{}] \}_{i=1}^N; ~ \LevelOneGtParamsguess^i \equiv \{\gtparams{1}, \gtparams{2}, \dots, \gtparams{N}\} \forall \playeridx{} \in [\numplayers{}]
% \end{equation}
% We define $\LevelOneGtParamsguess$ so that each $\LevelOneGtParamsguess^\playeridx$
% Let $\LevelOneGtParamsguess$ be \emph{homogeneous} level-2 inverse game parameter which is equivalent to a level-1 inverse 
% \begin{equation}
% \LevelOneGtParamsguess \equiv \{\LevelOneGtParamsguess^i\}_{i=1}^N; ~ \LevelOneGtParamsguess^i \equiv \{\gtparams{1}, \gtparams{2}, \dots, \gtparams{N}\} \forall \playeridx{} \in [\numplayers{}]
% \end{equation}
% }
% $\LevelOneGtParamsguess \equiv \{\LevelOneGtParamsguess^i\}_{i=1}^N$, where $\LevelOneGtParamsguess^i \equiv \{\gtparams{1}, \gtparams{2}, \dots, \gtparams{N}\}$ for all $i \in [N]$. 
Let $\bar{\Theta}^*$ denote the optimal homogeneous level-2 parameter minimizing $\loss(\LevelOneGtParamsguess)$. %The above analysis suggests that 
Then, the optimal level-2 inference solution $\GtParamsguess{*}$ always renders $\loss(\GtParamsguess{*})\le \loss(\LevelOneGtParamsguess^*)$.
Moreover, we establish the following upper and lower bounds on $\loss(\LevelOneGtParamsguess^*)$, assuming observation data is generated under the level-2 model.

% \begin{equation}
%     \loss(\LevelOneGtParamsguess) = \sum_{i=1}^N\sum_{t=1}^T \| G(\xstate{t}{i}(\LevelOneGtParamsguess)) -  \obs{t}{i} \|_2
% \end{equation}
% We denote by the minimizer of $L_1(Q)$ as $\hat{Q}^*$. By leveraging the structure of the matrix format KKT conditions, we can derive the following performance bounds of level-1 inference
\begin{proposition}\label{prop:bounds}
    Let $\GtParams{}^* \equiv \{\theta^{i*}\}_{i=1}^N$ be the level-1 ground truth parameters %parameterizing
    of all agents'
    % each agent's 
    objectives. Using these level-1 parameters, for each $i\in[N]$, we construct homogeneous level-2 parameters $\check{\Theta}^i\equiv \GtParams{}^*$,
    % \{ \theta^{j*}  \}_{j=1}^N$, 
    and define $\check{\Theta}\equiv \{\check{\Theta}^i\}_{i=1}^N$. %and $\{\theta^{i*}\}_{i=1}^N$. %\david{i'm not sure i fulloy this. also it is weird that these variables don't all arise in the eqns below} 
    %are the ground truth parameters for each agent's own objective.
    Denote by $\sigma_{\min}(\cdot)$ and $\sigma_{\max}(\cdot)$ the smallest and largest singular values, respectively. Define $\{\secondgtparams{i*}\}_{i=1}^N$ as in Proposition~\ref{prop:nonconvex main text}. The minimum loss of the level-1 inference can then be upper bounded by
    \begin{equation}
        \loss( \LevelOneGtParamsguess^* ) \!\le\! \frac{1}{2}\sum_{i=1}^N \left(\!\frac{\sigma_{\max}(M( \secondgtparams{i*} ) \!-\! M( \check{\Theta} ))}{\sigma_{\min}(M( \constructedlevelonesol ))} \|\LQtotalvariable^i(\secondgtparams{i*})\|_2\!\right)^2.
    \end{equation}
    In addition, let $\check{M}\equiv(\sum_{i=1}^N \frac{1}{N}(M(\secondgtparams{i*} ))^{-1})^{-1}$, and let $E^i $ be a matrix which selects the rows of the vector $\LQtotalvariable^i(\secondgtparams{i*})$ corresponding to agent $i$'s control. Then, we construct the lower bound %\david{is it simpler to just skip to (21)?}
    \begin{align}    
        \loss(\LevelOneGtParamsguess^*) 
        \ge  \frac{1}{2}\sum_{i=1}^N \left(\frac{\|E^i(M( \secondgtparams{i*} ) - \check{M} )\LQtotalvariable(\secondgtparams{i*})\|_2}{\sigma_{\max}( \check{M} )}\right)^2.
        % \ge\sum_{j=1}^N \frac{\|E^j\check{M} \LQtotalvariable^j\|_2}{\sigma_{\max}(M( \GtParamsguess{j} ))}.\nonumber
    \end{align}
\end{proposition}
% \hmzh{long sentence, not clear to me}
\noindent 
% \david{the notation is getting excessive. At the very least, you need to provide a complete table of notation in the supplementary material, and give an explicit pointer to that table here.}
% \todo{Make and cite the table with \textbf{every} symbol that exists in the paper.}
Proposition~\ref{prop:bounds} shows that the heterogeneity of the observed agents' ground truth level-2 parameters
% \todo{introduce a function that measures heterogeneity, and explain intutitively what it means/the difference with homogeneity}
$\{\secondgtparams{i*}\}_{i=1}^N$, which is measured by $\sum_{i=1}^N\|M( \secondgtparams{i*} ) - \check{M}\|_2$, affects the lower bound of the level-1 inference loss. In \cref{fig:lq-experiment}, we empirically demonstrate that level-2 inference in LQ settings, solved using the gradient-based inverse game solver introduced in the following sections, is capable of handling greater heterogeneity in agents’ estimates of each other’s objectives, whereas the performance of level-1 inference significantly deteriorates under similar conditions.

\section{Mixed Complementarity Transcription for Level-2 Inverse Games}
\label{ssec:forward-param-games-as-mcps}
% In this section, we derive KKT conditions for constrained Nash games under the level-1 assumption that all agents share the game parameters $\GtParams{}$. 
%We define a \emph{level-2 game} in which each agent $i$ knows its own true objective parameter $\gtparams{i}$ and maintains estimates $\params{i,-i}\equiv \{\params{i,1}, \cdots, \params{i,i-1}, \params{i,i+1}, \cdots, \params{i,N}\}$ of other agents' objectives, collectively denoted as $\GtParams{i} \equiv \{\gtparams{i}, \params{i,-i}\}$. Each agent independently computes an LGNE for its hypothesized game parameterized by $\GtParams{i}$, selects its control component from this equilibrium as its rational decision, and subsequently interacts with others.

Next, we consider a more general setting: nonlinear dynamics and non-quadratic costs.
% \hmzh{
% In this section, we introduce a method to solve level-2 inverse games. 
First, we transcribe the level-2 game into a mixed complementarity problem (MCP), which we solve using an off-the-shelf differentiable MCP solver.
Then, we infer the level-2 parameters from the data using the implicit function theorem to evaluate the gradients of the MCP solution with respect to the level-2 parameters, and the chain rule to evaluate the gradient of the loss in \cref{eq:level-2 log-likelihood-inverse-objective} with respect to those parameters.
% and minimize the loss $\loss( \GtParamsguess{})$ described in \cref{eq:level-2 log-likelihood-inverse-objective}.
% }

% , and introduce our level-2 inference algorithms. 
% \hmzh{As described in \cref{eq:level-2-inverse-game}, solving a level-2 inverse game requires solving a level-2 game.\jingqi{Jingqi: we need many calls of the level-2 games, so we can say "solving level-2 games". Actually, this first sentence is not as clear as the next sentence. I think we can remove it. }
In a level-2 game, each player plays a hypothesized level-1 game \cref{eq:level-2-inverse-game-nash-constraint}, $\Gamma(\GtParamsguess{i})$, and 
each resulting hypothesized LGNE $(\hypothesizedStateTraj{\playeridx}, \hypothesizedCtrlTraj{\playeridx}) \in \Gamma(\GtParamsguess{\playeridx})$ can be computed independently.
% For brevity, and because player $\playeridx{}$ is clear from context, we denote the $(\hypothesizedStateTraj{\playeridx}, \hypothesizedCtrlTraj{\playeridx})$ equivalently as
% \begin{equation}
% \label{eq:simplified-hypothesized-trajectories-for-i}
% \hypothesizedStateTraj{} 
% % \equiv \hypothesizedStateTraj{\playeridx} 
% \equiv\{\StateTraj{i,j}\}_{j=1}^{\numplayers{}} \text{ and } \hypothesizedCtrlTraj{} 
% % \equiv \hypothesizedCtrlTraj{\playeridx} 
% \equiv \{ \CtrlTraj{i,j}  \}_{j=1}^{\numplayers{}}.
% \end{equation} 
% }
% As each hypothesized level-1 game $\Gamma(\GtParamsguess{i})$ can be solved independently, we simplify the notation }
% \hmzh{For each agent $\secondidx{} \in [\numplayers{}]$ playing hypothesized game $\Gamma(\GtParamsguess{\playeridx})$}, \jingqi{The first sentence is confusing, it is hard to distinguish player j plays player i's hypothesized game, or player i hypothsized that player j acts in player i's hypothesized games
Consider agent $\playeridx$'s hypothesized game, $\Gamma(\GtParamsguess{i})$.
% Within each agent $i$'s hypothesized game $\Gamma(\GtParamsguess{i})$, 
For each agent $\secondidx \in [\numplayers{}]$, 
we associate inequality constraints 
% \todo{Does this writing look good and read clearly}?
% \david{need to watch indexing more carefully. x, u need indices no? what about i indices for multipliers?} \todo{g and h are defined over the whole trajectories. Was that the concern? x and u for the hypothesized game should have indices indicating their hypothetical nature.} \todo{These capital variables should each match those in (12) and be different depending on different hyp. games.}
% \todo{refer to $X^i, U^i$ as defined in \cref{eq:hypothesized-trajectories-for-i}}
$0 \le g^j(\hypothesizedStateTraj{\playeridx},\hypothesizedCtrlTraj{\playeridx})$ 
% $0 \le g^j(\StateTraj{},\CtrlTraj{})$ 
with multipliers $\hypinequalityConstraintLagrange{\playeridx,j}\ge 0$, and equality constraints 
$0 = h^j(\hypothesizedStateTraj{\playeridx},\hypothesizedCtrlTraj{\playeridx})$ 
% $0 = h^j(\StateTraj{},\CtrlTraj{})$
with multipliers $\hypequalityConstraintLagrange{\playeridx,j}$. 
Like in the LQ settings, at each stage $t$, we introduce multipliers for dynamics constraints  ($\hypdynamicsLagrange{t}{\playeridx,j}$}) and 
% is introduced. A multiplier 
initial state constraint $0=\hypxstate{1}{\playeridx,j} - \xstate{\textrm{init}}{j}$ ($\hypdynamicsInitLagrange{\playeridx,j}$). 
% $\hypdynamicsInitLagrange{\playeridx,j}$ is also introduced for the initial state constraint $0=\hypxstate{1}{\playeridx,j} - \xstate{\textrm{init}}{j}$. 
% As before, w
We define $\hypInequalityConstraintLagrange{\playeridx} \equiv \{\hypinequalityConstraintLagrange{\playeridx,j}\}_{j=1}^N$
$\hypEqualityConstraintLagrange{\playeridx} \equiv \{\hypequalityConstraintLagrange{\playeridx,j}\}_{j=1}^{\numplayers{}}\in\mathbb{R}^{\numequalitylagrange{}}$, and $\hypDynamicsInitLagrange{\playeridx} \equiv \{\hypdynamicsInitLagrange{\playeridx,j}\}_{j=1}^N\in \mathbb{R}^{\numconstlagrange{}}$. 
%Moreover, we denote by $\dynamicsLagrange{}{i}\equiv \{\dynamicsLagrange{t}{i}\}_{t=1}^{T-1} $ and $\DynamicsLagrange{}\equiv\{\dynamicsLagrange{}{i}\}_{i=1}^{N}$. 
%Let $\numequalitylagrange{}$ and $\numconstlagrange{}$ be the total number of equality and inequality constraints, respectively. 
% \todo{Add hats to duals.}
% \jingqi{Jingqi: Should we have hats on the duals?} \david{If hat is a consistent modifier for a hypothesized game, then yes. 
% \david{Put hats on z as well.}
We compactly store all variables as
\begin{equation}\label{eq:z}
    \totalvariable^{\playeridx} \equiv 
    [\hypothesizedStateTraj{\playeridx}, \hypothesizedCtrlTraj{\playeridx},
    % [\StateTraj{}, \CtrlTraj{},
    \hypDynamicsLagrange{\playeridx},\hypdynamicsInitLagrange{\playeridx},\hypequalityConstraintLagrange{\playeridx}, \hypinequalityConstraintLagrange{\playeridx} ] \in \mathbb{R}^{\horizon{}(\numstates{}  + \numctrls{} +  \numstates{} ) + \numequalitylagrange{} + \numconstlagrange{}  }.
\end{equation} 
% Under agent $i$'s estimate $\secondgtparams{i}$ of all agents' objectives, t
In agent $i$'s hypothesized game $\Gamma(\GtParamsguess{i})$, the Lagrangian for agent $j$ is defined as follows:  
% \david{can probably use align and get a little cleaner} %\david{watch i vs j}
% \david{why no eqn number here or below? also please fix awk spacing. i bet claude/gpt can fix instantly fwiw}
% \david{please make this more compact. it is a waste of space now}
\begin{equation}
%  \begin{aligned}
%     \lagrangian{\playeridx,j}(\totalvariable^{\playeridx}; 
%     % \GtParamsguess{i} 
%     \params{\playeridx,\secondidx})
%     =& \costfn{j}(\hypothesizedStateTraj{\playeridx}, \hypothesizedCtrlTraj{\playeridx}; \params{\playeridx,\secondidx}) 
%     + \hypequalityConstraintLagrange{\playeridx,j}{}^\top \equalityconstraint{j}(\hypothesizedStateTraj{\playeridx}, \hypothesizedCtrlTraj{\playeridx})
%     % (\StateTraj{}, \CtrlTraj{})
%     + \hypdynamicsInitLagrange{\playeridx,j}{}^\top(\hypxstate{1}{\playeridx{}, j} - \xstate{\textrm{init}}{j}) + \\
%     &\sum_{t=1}^{T-1}\hypdynamicsLagrange{t}{\playeridx,j}{}^\top (\hypxstate{t+1}{\playeridx, j} - \dynamics{j}(\hypxstate{t}{\playeridx, j}, \hypctrl{t}{\playeridx,j}))
%     - \hypinequalityConstraintLagrange{\playeridx,j}{}^\top \inequalityconstraint{j}(\hypothesizedStateTraj{}, \hypothesizedCtrlTraj{}).
%     % (\StateTraj{}, \CtrlTraj{}). 
% \end{aligned}
 \begin{aligned}
    \lagrangian{\playeridx,j}(\totalvariable^{\playeridx}; 
    % \GtParamsguess{i} 
    \params{\playeridx,\secondidx})
    =& \costfn{j}(\hypothesizedStateTraj{\playeridx}, \hypothesizedCtrlTraj{\playeridx}; \params{\playeridx,\secondidx}) 
    + \hypequalityConstraintLagrange{\playeridx,j}{}^\top \equalityconstraint{j}(\hypothesizedStateTraj{\playeridx}, \hypothesizedCtrlTraj{\playeridx}) \\
    % (\StateTraj{}, \CtrlTraj{})
    &+ \sum_{t=1}^{T-1}\hypdynamicsLagrange{t}{\playeridx,j}{}^\top (\hypxstate{t+1}{\playeridx, j} - \dynamics{j}(\hypxstate{t}{\playeridx, j}, \hypctrl{t}{\playeridx,j})) \\
    &+ \hypdynamicsInitLagrange{\playeridx,j}{}^\top(\hypxstate{1}{\playeridx{}, j} - \xstate{\textrm{init}}{j})
    - \hypinequalityConstraintLagrange{\playeridx,j}{}^\top \inequalityconstraint{j}(\hypothesizedStateTraj{}, \hypothesizedCtrlTraj{}).
    % (\StateTraj{}, \CtrlTraj{}). 
\end{aligned}
\end{equation}
% \begin{subequations}
%     \begin{align}
%         \lagrangian{j}(\totalvariable; \GtParamsguess{i})
%     =~& \costfn{j}(\StateTraj{}, \CtrlTraj{}; \GtParamsguess{i})  \\
%     &\quad + {\equalityConstraintLagrange{j}}^\top \equalityconstraint{j}(\StateTraj{}, \CtrlTraj{}) \\
%     &\quad + {\dynamicsInitLagrange{j}}^\top(\xstate{1}{j} - \xstate{\textrm{init}}{j})  \\ 
%     &\quad + \sum_{t=1}^{T-1}{\dynamicsLagrange{t}{j}}^\top (\xstate{t+1}{j} - \dynamics{j}(\xstate{t}{j}, \ctrl{t}{j})) \\
%     &\quad - {\inequalityConstraintLagrange{j}}^\top \inequalityconstraint{j}(\StateTraj{}, \CtrlTraj{}). 
%     \end{align}
% \end{subequations}
% \david{need to explain dropped superscripts and subscripts}
Under an appropriate constraint qualification (e.g., the linear independence constraint qualification, cf. \cite{wright1999numerical}),
% We assume certain standard assumptions hold, namely the linear independent constraint qualification (LICQ) \cite{wright1999numerical}.
% Under the LICQ assumption,
a hypothesized LGNE satisfies the following Karush-Kuhn-Tucker (KKT) conditions
\begin{equation}
    \forall j \in [N] \left\{\begin{aligned}
        & 0=\nabla_{(\hypothesizedStateTraj{\playeridx}, \hypothesizedCtrlTraj{\playeridx{},\secondidx{}})}
        % (\StateTraj{}, \CtrlTraj{j})} 
        \lagrangian{\playeridx,j}
        (
        %\StateTraj{}, \CtrlTraj{}, \DynamicsLagrange{}, \DynamicsInitLagrange{},\DynamicsLagrange{}, \InequalityConstraintLagrange{} 
        \totalvariable^\playeridx;
        \params{\playeridx,\secondidx} ),\\
        % \GtParamsguess{i} )\\
        & 0 = \equalityconstraint{j}(\hypothesizedStateTraj{\playeridx}, \hypothesizedCtrlTraj{\playeridx}), \\
        % (\StateTraj{}, \CtrlTraj{}) \\
        & 0 = \hypxstate{1}{\playeridx, j} -
        % \xstate{1}{j} -
        \xstate{\textrm{init}}{j},\\
        % & 0 = \xstate{t+1}{j} - \dynamics{j}(\xstate{t}{j}, \ctrl{t}{j}),\ t\in[T-1] \\
        & 0 = \hypxstate{t+1}{\playeridx, j} - \dynamics{j}(\hypxstate{t}{\playeridx, j}, \hypctrl{t}{\playeridx, j}),\ t\in[T-1], \\
        & 0 \le \inequalityconstraint{j} 
        (\hypothesizedStateTraj{\playeridx}, \hypothesizedCtrlTraj{\playeridx})
        % (\StateTraj{}, \CtrlTraj{}) 
        \perp \hypinequalityConstraintLagrange{\playeridx,j}\ge 0.
    \end{aligned}\right.
\end{equation}
Following standard procedures, cf. \cite[Ch. 1]{facchinei2003finite} and \cite[Eq. (50)]{SmoothGameTheory}, we can transcribe these KKT conditions as a mixed complementarity problem (MCP), for which efficient off-the-shelf solvers exist \cite{dirkse1995path}. To this end,
% Motivated by the numerical efficiency of existing solvers such as PATH \cite{dirkse1995path} for solving such mixed complementarity problems \cite[\S9.4.2]{facchinei2003finite}, we transform the above KKT conditions into mixed complementarity conditions by following a similar process as in 
% described in
% \citet[Eq. (50)]{SmoothGameTheory}.
% and \citet{SmoothGameTheory}}, 
we define two vector functions 
% \david{why not both out of line? should fit} \todo{Is this what was meant? Make it a column vector and side by side.}
$F_\textrm{ineq}(\totalvariable^{\playeridx};\GtParamsguess{i})\equiv[g^{1\top}, \cdots, g^{N\top}]^\top$
% \begin{equation*}
% F_\textrm{ineq}(\totalvariable^{\playeridx};\GtParamsguess{i})\equiv[g^{1\top}, \cdots, g^{N\top}]^\top
% \end{equation*} 
and 
% $F_\textrm{eq}(\totalvariable;\GtParamsguess{i})$: %\david{watch bolding of z, and check that transposes are always the same symbol $\top$ vs $\T$}: %\david{format better}
\begin{equation*}
    % F_\textrm{ineq}(\totalvariable^{\playeridx};\GtParamsguess{i})\equiv\left[\begin{array}{c}g^{1\top} \\ \vdots \\ g^{N\top}\end{array}\right],
    % \text{ and }
    F_{\textrm{eq}}(\totalvariable^{\playeridx};\GtParamsguess{i}) \equiv\begin{bmatrix}\begin{bmatrix}
        \nabla_{(\hypothesizedStateTraj{\playeridx}, 
        \hypothesizedCtrlTraj{\playeridx{},\secondidx{}}
        % \hypothesizedCtrlTraj{\playeridx, \secondidx}
        )} \lagrangian{\playeridx,j} \\
        % \nabla_{(\StateTraj{}, \CtrlTraj{j})} \lagrangian{j} \\ 
        \equalityconstraint{j}  \\
        \hypxstate{1}{\playeridx, j} - \xstate{\textrm{init}}{j}\\
        % \xstate{1}{j} - \xstate{\textrm{init}}{j}\\
        \big[[\hypxstate{t+1}{\playeridx, j} - \dynamics{j} (\hypxstate{t}{\playeridx, j}, \hypctrl{t}{\playeridx, j})\big]_{t=1}^{T-1}
    \end{bmatrix}_{j=1}^N\end{bmatrix},
\end{equation*}
and we aim to find a solution $\totalvariable^*$ for the MCP % the mixed complementary conditions:
\begin{equation}\label{eq: MCP conditions nonlinear}
    \left\{\begin{aligned}
        & 0 = F_{\textrm{eq}}(\totalvariable^{\playeridx}; \GtParamsguess{i}) \\
        & 0 \le  \hypinequalityConstraintLagrange{\playeridx}  \ \  \perp \ \ F_{\textrm{ineq}}(\totalvariable^{\playeridx}; \GtParamsguess{i}) \ge 0. 
    \end{aligned}\right.
\end{equation}
% Leveraging the PATH solver, we can efficiently compute an LGNE. 
% \todo{why do we need these conditions and what do they mean?}
% \david{PATH is not differentiable directly. ParametricMCPs built the IFT bit and is actually solver independent. It should be mentioned in the text.}
Moreover, under certain standard regularity conditions such as strict complementarity  \cite{wright1999numerical}, following \cite{liu2023learning} we can apply the implicit function theorem \cite{krantz2002implicit} to the above mixed complementarity conditions and use the chain rule to obtain the gradient of $\totalvariable^{\playeridx*}$ with respect to $\GtParamsguess{i}$, i.e., $\nabla \totalvariable^{\playeridx*}(\GtParamsguess{i})$.
In our work, we use the ParametricMCPs.jl library \cite{ParametricMCPs} to accomplish this gradient evaluation. 
% \todo{find a cite for this?}
%With this gradient available, we can derive the loss gradient $\nabla \loss(\GtParamsguess{})$:%With this in hand, we can derive the loss gradient $\nabla \loss(\GtParamsguess{}) $:
% , through automatic differentiation \cite[Ch. 8]{wright1999numerical}. 
% Applying chain rule to the loss function $\loss(\GtParamsguess{})$ in \eqref{eq:level-2 log-likelihood-inverse-objective}, we can derive a gradient estimator for efficiently minimizing the nonconvex loss $\loss( \GtParamsguess{} )$:
% \begin{align}\label{eq:loss gradient}
% \nabla \loss(\GtParamsguess{}) & = \sum_{i=1}^N  \nabla \loss(\GtParamsguess{i})\\
%     \nonumber  = \sum_{i=1}^N & \sum_{t=1}^T  \left(\frac{\partial \xstate{t}{i,i*}}{\partial\GtParamsguess{i}}\right)^\top   \! \left( \frac{ \partial G(\xstate{t}{i,i*})}{\partial \xstate{t}{i,i*}} \right)^\top \mathbb{I}_{\obs{t}{\playeridx}}  \Big( G( \xstate{t}{i,i*}) - \obs{t}{i} \Big)
% \end{align}
After reshaping $\GtParamsguess{}$ as a vector $\thetavec$, we obtain the gradient of $\loss(\GtParamsguess{})$ with respect to an arbitrary element $\thetavec_j$ via the chain rule:% through automatic differentiation \cite[Ch. 8]{wright1999numerical}:
% \begin{equation}\label{eq:loss gradient}
%     \frac{\partial \loss }{\partial \thetavec_j} = \sum_{i=1}^N\sum_{t=1}^T\frac{\partial \loss  }{\partial  G \Big(\hypxstate{t}{(i,i)*}\Big) } \cdot  \frac{\partial G\Big(\hypxstate{t}{(i,i)*}\Big) }{\partial  \hypxstate{t}{(i,i)*} } \cdot \frac{\partial \hypxstate{t}{(i,i)*}  }{\partial \thetavec_j  }
% \end{equation}
% \jingqi{
% \begin{equation}\label{eq:loss gradient}
% \begin{aligned}
%     \frac{\partial \loss }{\partial \thetavec_j} &= \sum_{i=1}^N\frac{\partial \loss  }{\partial  G \Big(\hypothesizedStateTraj{(i,i)*}, \hypothesizedCtrlTraj{(i,i)*} \Big) } \cdot  \\
%     \Bigg(& \frac{\partial G\Big(\hypothesizedStateTraj{(i,i)*}, \hypothesizedCtrlTraj{(i,i)*}\Big) }{\partial  \hypothesizedStateTraj{(i,i)*} } \cdot \frac{\partial \hypothesizedStateTraj{(i,i)*}  }{\partial \thetavec_j  }+
%     \frac{\partial G\Big(\hypothesizedStateTraj{(i,i)*}, \hypothesizedCtrlTraj{(i,i)*}\Big) }{\partial  \hypothesizedCtrlTraj{(i,i)*} } \cdot \frac{\partial \hypothesizedCtrlTraj{(i,i)*}  }{\partial \thetavec_j  }\Bigg)
% \end{aligned}
% \end{equation}
% }
\begin{equation}
\label{eq:loss gradient}
\frac{\partial \loss}{\partial \thetavec_j} \!\!=\!\!
\sum_{i=1}^N 
% \frac{\partial \loss}{\partial G\!\Big(\hypothesizedStateTraj{(i,i)*}\!,\! \hypothesizedCtrlTraj{(i,i)*}\Big)}
\frac{\partial \loss}{\partial G(\zeta^{i*})}
 % \cdot 
\Bigg(
\frac{\partial G}{\partial \hat\zeta^{i*}_{\hypothesizedStateTraj{}}}
\frac{\partial \hat\zeta^{i*}_{\hypothesizedStateTraj{}}}{\partial \thetavec_j}
+
\frac{\partial G}{\partial \hat\zeta^{i*}_{\hypothesizedCtrlTraj{}}}
\frac{\partial \hat\zeta^{i*}_{\hypothesizedCtrlTraj{}}}{\partial \thetavec_j}
\Bigg),
% _{\displaystyle \frac{\partial G\!\Big(\hypothesizedStateTraj{(i,i)*}, \hypothesizedCtrlTraj{(i,i)*}\Big)}{\partial \thetavec_j}}
\end{equation}
where $\hat\zeta^{i*} = (\hat\zeta^{i*}_{\hypothesizedStateTraj{}}, \hat\zeta^{i*}_{\hypothesizedCtrlTraj{}})$
% = (\hypothesizedStateTraj{(i,i)*}, \hypothesizedCtrlTraj{(i,i)*})$ 
is a tuple containing
% $\hypothesizedStateTraj{(i,i)*}$ and $\hypothesizedCtrlTraj{(i,i)*}$ represent
player $\playeridx$'s own LGNE state and control trajectories in its own hypothesized game $\Gamma(\GtParamsguess{i})$.
% 's , respectively.
Using \cref{eq:loss gradient}, we solve the level-2 inverse game by minimizing $\loss(\GtParamsguess{})$ in \cref{eq:level-2 log-likelihood-inverse-objective} using gradient descent.
% to solve the level-2 inverse game.
The additive structure of \cref{eq:loss gradient} enables efficient parallel computation.
% over each agent.
\vspace{-12pt}

% \todo{Describe what $\loss$ is and how we solve using these gradients?}

% \hmzh{Shorten + repeated below.}
% We note that the gradient estimator in equation \eqref{eq:loss gradient} can be efficiently computed in parallel over different agents $j\in[N]$ due to the additive structure. 
% We  leverage this gradient estimator and formally present the inverse problem associated with this level-2 game in the following section. 

% {
% \color{red}
% After reparameterizing $\GtParamsguess{}$ as a vector $\thetavec$, we obtain the gradient of $\loss(\GtParamsguess{})$ with respect to an arbitrary element $\thetavec_j$:
% \begin{equation}
%     \frac{\partial \loss }{\partial \thetavec_j} = \frac{\partial \loss  }{\partial  (G (\xstate{t}{i,i*})) } \cdot  \frac{\partial G(\xstate{t}{i,i*}) }{\partial  \xstate{t}{i,i*} } \cdot \frac{\partial \xstate{t}{i,i*}  }{\partial \thetavec_j  }
% \end{equation}

% }

\input{figs/lq_loss}

\input{figs/mismatched_forward_LCs}

\input{figs/LC_figure}

\section{An Efficient Level-2 Inverse Game Algorithm}

We solve the level-2 inverse game using gradient descent on the loss $\loss(\GtParamsguess{})$. 
At each iteration, we compute gradients via implicit differentiation of the MCP solution described in \cref{ssec:forward-param-games-as-mcps} using \cref{eq:loss gradient}, and update the parameter estimate $\GtParamsguess{}$. 
These gradient evaluations can be performed in parallel across agents, and the procedure iterates until convergence or a maximum number of iterations is reached.

Our method supports both offline and online inference.
In the offline setting, we estimate parameters from historical observations $\ObsTraj{}$ by minimizing $\loss(\GtParamsguess{})$, as illustrated in \cref{fig:lq-experiment}.
In the online setting (\cref{fig:lane changing online}), we apply the method within a receding-horizon framework \cite{peters2023learning} to update estimates from streaming data, a more challenging inference setting evaluated in the following section.
To ensure tractability, we assume parameters remain fixed over the observation horizon and handle time variation using receding-horizon updates.

%% file: figs/lq_loss.tex
\begin{figure}[!b]
    \centering
    \includegraphics[width=0.9\linewidth]{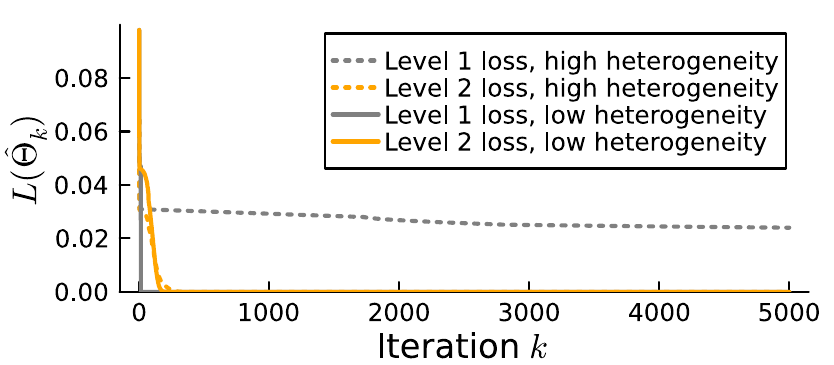}\vspace{-0.4em}
    \caption{
    Level-2 inference achieves lower inverse game loss values than level-1 inference when observed agents' estimates of others' objectives differ significantly. 
    %\david{Can we use the proper math symbols on the vertical axis? would also rm the plot title, and make the plot full column width} 
    We use our method to solve the level-2 inverse game, with further details provided in Supplement C \cite{khan2025agentsthinkdolevel2}. 
    % \todo{TODO: \cite{khan2025agentsthinkdolevel2}}
    % the supplementary material. 
    %When agents’ estimates of the game objectives differ significantly, for instance, in the dotted curve scenario characterized by a large value of $\sum_{j=1}^N \|M(\GtParamsguess{j}) - \check{M}\|_2$, level-2 inference performance much better than level-1 inference, which struggles to minimize the loss function $\loss(\LevelOneGtParamsguess^*)$. %Although level-1 inference performs relatively well when agents’ true parameters $\GtParamsguess{i}$ are similar (solid grey), level-2 inference consistently achieves low loss values and maintains strong performance.
    %\todo{Inference results for an LQ game. \emph{Level-2 inference on data generated from fictitious play on this game is able to identify the correct (optimal) parameters, whereas level-1 inference is not.}}
    % \vspace{-12pt}
    % \todo{add description alt text}
    % \Description{Level-2 inference achieves lower inverse game loss values than level-1 inference when observed agents' estimates of others' objectives differ significantly.}
    }
    \label{fig:lq-experiment}
\end{figure}

%% file: figs/mismatched_forward_LCs.tex
\begin{figure*}[!t]
\centering

\begin{minipage}[t]{0.58\textwidth}
\vspace{0pt}
\centering
\subfloat[Aligned parameters; lane change succeeds safely.\label{fig:forward-lc-mismatched-safe}]{
    \includegraphics[width=0.95\linewidth]{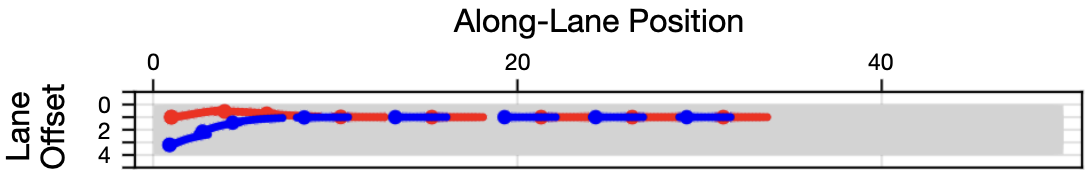}
}\\[7pt]
\subfloat[Mismatched parameters, deadlocked lane change.\label{fig:forward-lc-mismatched-deadlock}]{
    \includegraphics[width=0.95\linewidth]{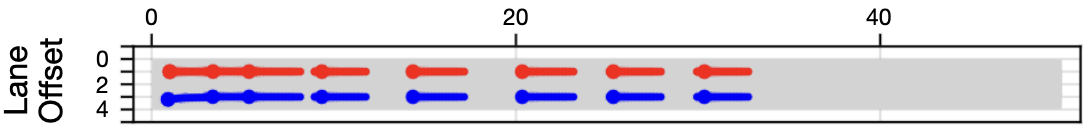}
}
\end{minipage}\hfill
\begin{minipage}[t]{0.40\textwidth}
\vspace{0pt}
\centering
\subfloat[Time to successful lane change.\label{fig:lane-change-times}]{
    \includegraphics[width=\linewidth]{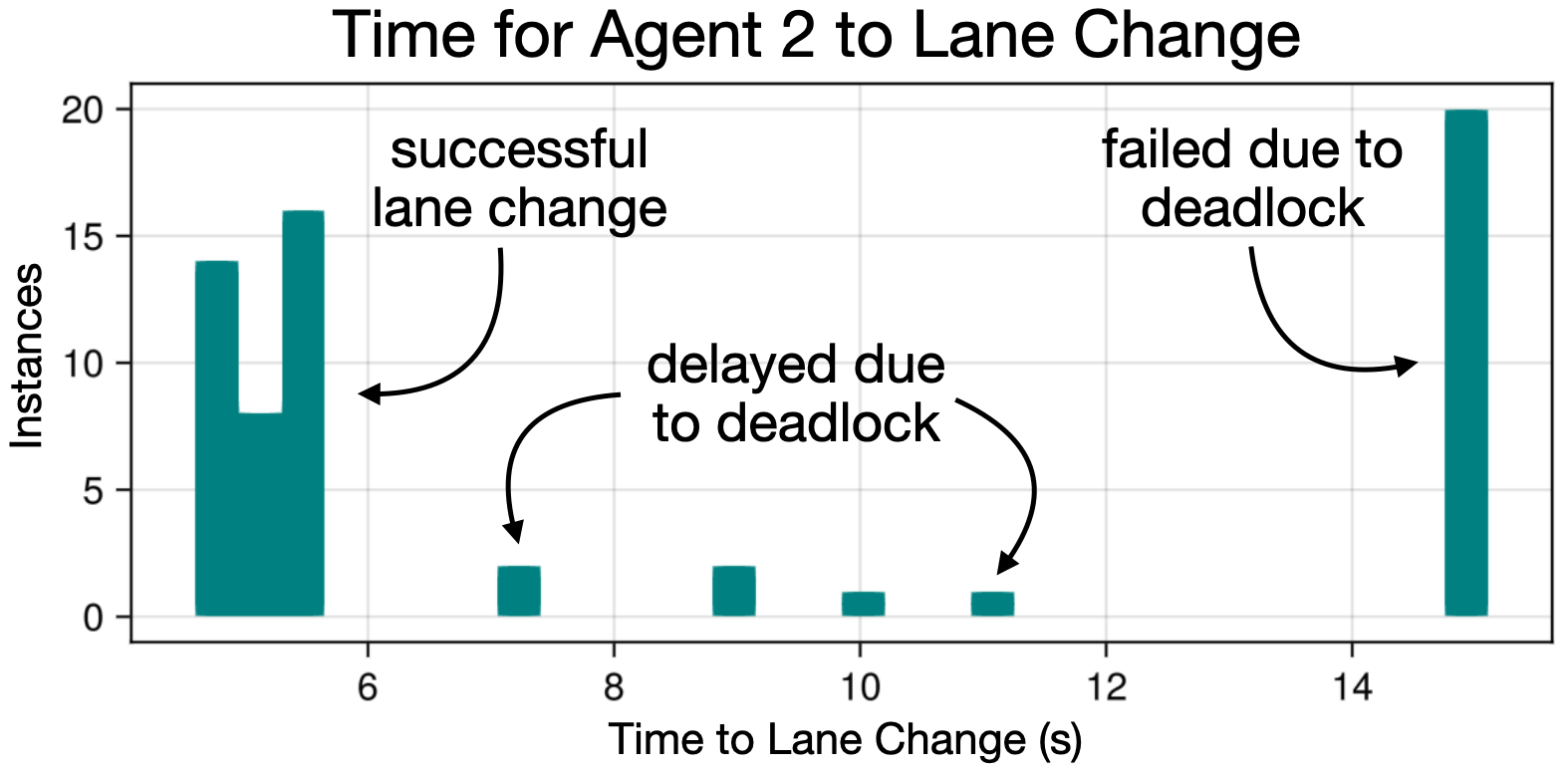}
}
\end{minipage}

\caption{\label{fig:forward-lc-sweep}
\textit{Fictitious play between two agents can lead to unsuccessful interactions under mismatched objective estimates.}
(a) Aligned beliefs yield a safe lane change.
(b) Mismatched beliefs cause deadlock.
(c) Histogram of lane-change completion time across 64 simulations with varying belief mismatch.
}
\vspace{-1em}
% \textit{Fictitious play between two agents frequently, but not always, leads to unsuccessful interaction.}

% Both agents wish to track the center of the top lane (\SI{1}{\meter}).
% When these agents incorrectly estimate each other's objectives, it can lead to a variety of both successful (\cref{fig:forward-lc-mismatched-safe}) and unsuccessful (\cref{fig:forward-lc-mismatched-deadlock}) interactions.
% We simulate 64 instances of fictitious play in which agents' estimates of each other's desired lane offset vary uniformly from \SI{0.5}{\meter} to \SI{4}{\meter}.
% \cref{fig:lane-change-times} shows the time it takes for agent~2 to successfully lane change in each instance.
% In many of these instances, agent~2 successfully changes its lane early in the simulation (\cref{fig:forward-lc-mismatched-safe}); this can occur safely, but it can also occur abruptly or aggressively due to mismatch (not shown).
% In other instances, agent~2's lane change is delayed or fails due to deadlock (\cref{fig:forward-lc-mismatched-deadlock}).
% }
% \label{fig:forward-lc-sweep}
\end{figure*}

%% file: figs/LC_figure.tex
\begin{figure*}[!ht]
\centering
\subfloat[Level-1 inference on the lane change in \cref{fig:front-figure}.\label{fig:level-1-front-lc-inference}]{
    \includegraphics[width=0.4\textwidth]{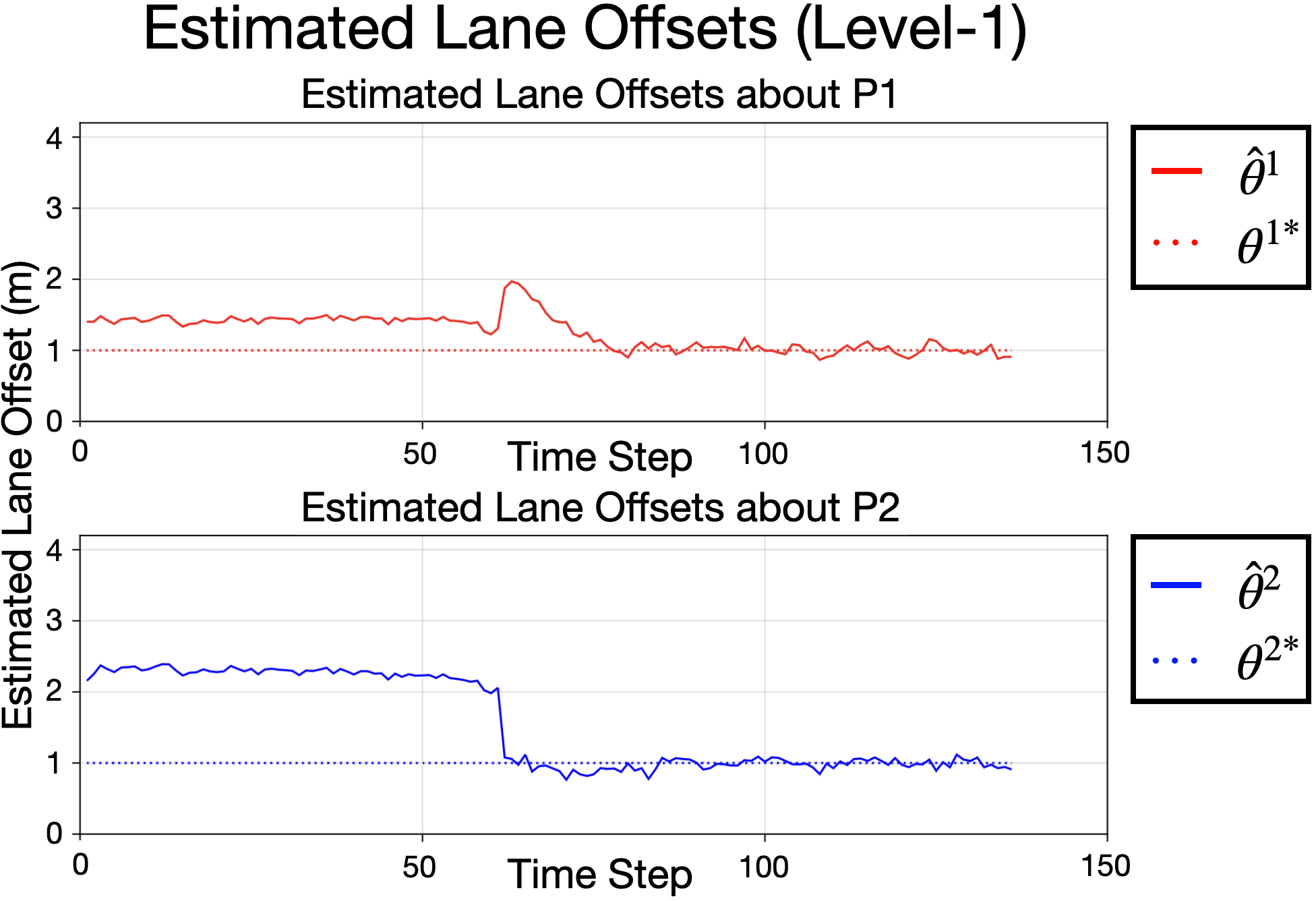}
}%\hfill
\subfloat[Level-2 inference on the lane change in \cref{fig:front-figure}.\label{fig:level-2-front-lc-inference}]{
    \includegraphics[width=0.4\textwidth]{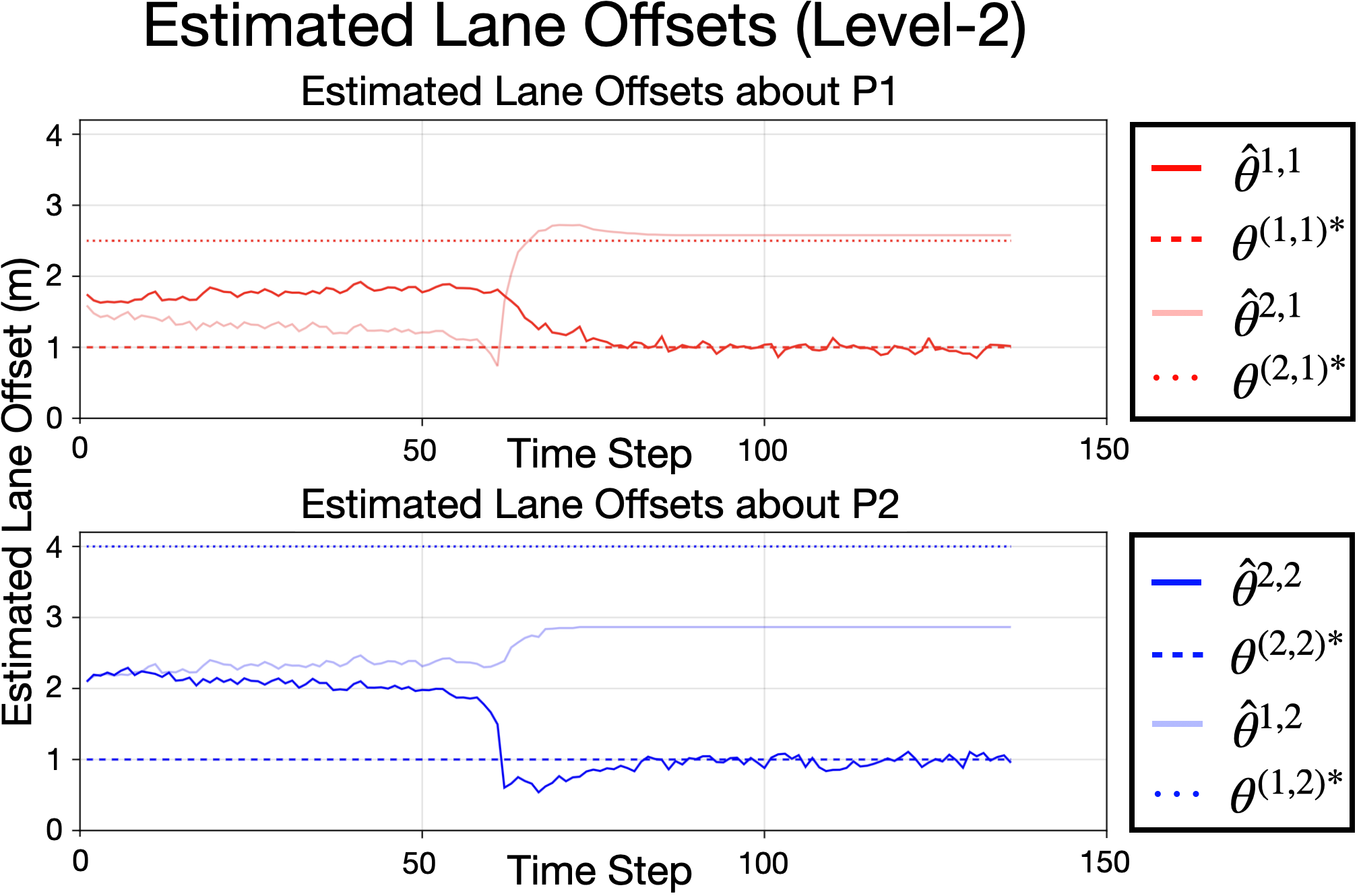}
}

\caption{
\textit{Online parameter estimates in the lane-change scenario.}
Level-1 (\cref{fig:level-1-front-lc-inference}) and level-2 (\cref{fig:level-2-front-lc-inference}) inference are applied to the interaction in \cref{fig:front-figure}. 
Level-2 inference recovers the mismatched belief parameters $\params{\playeridx{},\secondidx{}}$, while level-1 inference estimates only agents' objectives.
% \textit{Level-2 inference infers underlying mismatched objectives $\params{\playeridx{},\secondidx{}}$ from receding horizon gameplay data, while level-1 inference does not.}
% We run level-1 (\cref{fig:level-1-front-lc-inference}) and level-2 (\cref{fig:level-2-front-lc-inference}) inference on the lane change example depicted in \cref{fig:front-figure}, which initially deadlocks but eventually results in a successful lane change.
% Both level-1 and level-2 inference initially identify incorrect estimates for the agents' objective parameters but eventually converge to correct estimates for individual agent objectives during the simulation.
% Level-2 inference produces correct (mismatched) values for estimates of the other agent's objectives.
}
\vspace{-1em}
% \label{fig:lc-inference}
\label{fig:lane changing online}
\end{figure*}

%% file: include/sec5_experiments_new.tex
\section{Lane Change Experiments}
\label{sec:experiments}
% \label{ssec:running-example}
We evaluate our method on a two-vehicle lane-change scenario on a two-lane road (Figure~\ref{fig:front-figure}). 
Each vehicle follows planar double-integrator dynamics and seeks to reach 
a target lane while maintaining a desired velocity and avoiding collision. 
Each agent plans its strategy based on its own target lane parameter 
$\gtparams{i,i}$ and its estimate of the other agent’s target lane $\gtparams{i,-i}$. 
Each agent's objective penalizes lane tracking error, velocity tracking error, and 
control effort, and agents enforce a minimum-distance collision constraint. 
In this scenario, a third-party observer estimates parameters 
$\GtParamsguess{\playeridx{}} = \{ \params{\playeridx{}, \playeridx{}}, \params{\playeridx{},-\playeridx{}} \}$ 
for each agent by maximizing the likelihood of the observed trajectories $\ObsTraj{}$ under the resulting LGNE.
Full model details and hyperparameters are provided in Appendix~C~\cite{khan2025agentsthinkdolevel2}.

\subsubsection{Fictitious Play Can Cause Unsafe Behavior}

We consider a scenario in which both agents aim to switch to the lane offset $\nicefrac{l_w}{2}$ (the center of the top lane), but each maintains an incorrect estimate of the other agent’s target lane. 
As a result, the agents plan using different Nash games derived from their misaligned beliefs.
We select true parameters 
$\secondgtparams{*} = \{\secondgtparams{1*}, \secondgtparams{2*}\}$, where 
$\secondgtparams{1*} = \{\theta^{(1,1)*}, \theta^{(2,1)*}\}$ and 
$\secondgtparams{2*} = \{\secondgtparams{(2,1)*}, \secondgtparams{(2,2)*}\}$.
We then simulate fictitious play in game $\Gamma(\secondgtparams{*})$ in a receding-horizon manner.

% We select true parameters $\Theta^* = \{\{\theta_{(1,1)}^*, \theta_{(2,1)}^*\}, \{\theta_{(2,1)}^*, \theta_{(2,2)}^*\}\}$.
% Here $\theta_{(1,1)}^* = \theta_{(2,2)}^* = l_w/2$, and we sweep over a range of belief parameters $\theta_{(1,2)}^*, \theta_{(2,1)}^*$.
% We then simulate fictitious play in $\Gamma(\Theta^*)$ in a receding-horizon manner.

Figure~\ref{fig:forward-lc-sweep} shows representative outcomes. 
Aligned beliefs ($\gtparams{(1,2)*} = \gtparams{(2,1)*} = \nicefrac{l_w}{2}$) produce a successful lane change (\cref{fig:forward-lc-mismatched-safe}). 
Moderate belief mismatches can still yield successful lane changes, but some indicate aggressive maneuvers, as in the left mode of the histogram in \cref{fig:lane-change-times}. 
Larger mismatches often lead to deadlock (\cref{fig:forward-lc-mismatched-deadlock}), while intermediate cases eventually resolve after an initial deadlock (\cref{fig:front-figure}).

These results show that mismatched beliefs can produce unsafe or unsuccessful interactions and motivate the need for level-2 inference, which explicitly models agents’ estimates of one another’s objectives.

\subsubsection{Level-2 Inference Identifies Mismatched Beliefs}

\Cref{fig:lane changing online} compares level-1 
(\cref{fig:level-1-front-lc-inference}) and level-2 
(\cref{fig:level-2-front-lc-inference}) inference on the scenario in \cref{fig:front-figure}. 
Early in the simulation, both methods converge to similar incorrect 
parameter estimates near the road center at \SI{2}{\meter}, 
$\params{\playeridx{}} \approx \params{\playeridx{},\playeridx{}} 
\approx \params{-\playeridx{},\playeridx{}}$. 
The estimate $\params{2,1}$ gradually moves toward the first lane center 
at \SI{1}{\meter} until around time step 60. 
Over the interaction, the estimates converge to the true parameters, 
yielding $\params{\playeridx{}} \approx \params{\playeridx{},\playeridx{}} = 1$.

Level-2 inference additionally estimates the belief parameters 
$\params{1,2}$ and $\params{2,1}$, revealing that each agent initially 
believes the other aims for the road center, which explains the observed 
deadlock. Although this explanation is incorrect, the nonconvexity of the 
level-2 inference problem (\Cref{prop:nonconvex main text}) means that 
multiple local solutions are expected.

When agent 2 begins its lane change, the estimate $\params{2,1}$ quickly 
converges to the true mismatched value $\gtparams{(2,1)*}$. The estimate 
$\params{1,2}$ does not fully converge to $\gtparams{(1,2)*}$ because once 
the agents separate, the equilibrium becomes insensitive to these parameters 
(\cref{rmk:gradients}).
Nevertheless, level-2 inference correctly identifies 
that agent 1's behavior is driven by a belief that agent 2 prefers the bottom lane.
These results show that level-1 inference cannot capture behaviors arising 
from mismatched beliefs, while level-2 inference can explicitly recover 
agents’ estimates of one another’s objectives.

%% file: include/sec6_conclusion.tex
\section{Conclusion}

We introduce a level-2 inference framework that relaxes the common assumption in inverse game-theoretic approaches that agents share complete knowledge of one another's objectives. 
Our theoretical and empirical analyses demonstrate that level-1 inference produces significant prediction errors when agents maintain heterogeneous estimates of others' objectives.
We develop an efficient gradient-based algorithm that infers these misaligned beliefs by formulating level-2 dynamic games as differentiable mixed complementarity problems. 
Experiments on LQ games and a synthetic urban driving scenario demonstrate that the proposed method identifies misalignments that simpler level-1 inference approaches fail to detect.
Future work includes applying level-2 inference to strategic interactive planning, extending the framework to nonlinear stochastic games, and formally characterizing the observability of level-2 parameters. 
An open challenge is developing methods that reason about competing hypotheses explaining observed strategic interactions.

%% file: include/proof_details.tex
\section{Supplement A: A Compact Representation of the KKT Conditions for the LQ Game}
In this section, we derive a compact linear representation for the KKT system of the LQ game used in the proofs of Proposition \ref{prop:nonconvex main text} and Proposition \ref{prop:bounds}.

% In a $N$-agent, $T$-horizon linear-quadratic game, the $i$-th player considers the stage cost
% \begin{equation}
%     \ell_t^i(x,u) = \frac{1}{2} x_{t+1}^\top Q^i x_{t+1} + \frac{1}{2} u_t^\top R^i u_t
% \end{equation}
% and the whole horizon cost is defined as $L^i(x,u) = \sum_{t=0}^T \ell_t^i(x,u)$. The overall system dynamics, which include all players, is defined as
% \begin{equation}
%     x_{t+1} = A x_t + B u_t
% \end{equation}
% Denote by $x:= [x_t]_{t=0}^{T+1}$, $u:=[u_t^i]_{i=1,t=0}^{N, T}$, and $\lambda := [\lambda_t^i]_{i=1,t=0}^{N,T}$. Consider the Lagrangian
% \begin{equation*}
%     \mathcal{L}_t^i(x,u,\lambda): = x_{t+1}^\top Q^i x_{t+1} + \frac{1}{2} u_t^\top R^i u_t + \lambda_t^{i\top} (x_{t+1} - A x_t - B u_t)
% \end{equation*}
% and 
% \begin{equation}
%     \mathcal{L}^i(x,u,\lambda) : = \sum_{t=0}^T \mathcal{L}_t^i(x,u,\lambda)
% \end{equation}
% \begin{definition}
%     An open-loop policy $\pi^*: x_0 \to u$ is an open-loop Nash equilibrium if
%     \begin{equation}
%     \begin{aligned}
%         & L^i(x, \pi^{1*}(x_0), \cdots, \pi^{i}(x_0),\cdots, \pi^{N*}(x_0) ) \\ 
%         & \ge L^i(x,\pi^{1*}(x_0), \cdots, \pi^{i*}(x_0),\cdots, \pi^{N*}(x_0) ).
%     \end{aligned} 
%     \end{equation}
% \end{definition}
In an $N$-agent, $T$-horizon LQ game, the $i$-th agent considers the quadratic stage cost
\begin{equation}
    \label{eq:exp-lq-cost}
    \stagecostfn{t}{i}(\xstate{t}{}, \ctrl{t}{} ;\gtparams{i}) = \frac{1}{2} {\xstate{t}{}}^\top Q({\gtparams{i}}) \xstate{t}{} + \frac{1}{2} {\ctrl{t}{i}}^\top R^i \ctrl{t}{i}
\end{equation}
where the entries of the positive semidefinite matrix $Q({\gtparams{i}})$ are parameterized by the vector $\gtparams{i}$, and $R_t^i$ is a fixed positive definite matrix. 
%We define the total cost as $\costfn{\playeridx}(\StateTraj{}, \CtrlTraj{}; \gtparams{\playeridx})  = \sum_{t=1}^T \ell_t^i(x,u;\gtparams{i})$. 
Each agent $\playeridx$'s state dynamics are given by %, its dynamics $\dynamics{i}$ is
\begin{equation}
    \label{eq:exp-lq-dyn}
    \xstate{t+1}{i} = \LinearKKTA{}^{\playeridx} \xstate{t}{i} + \LinearKKTB{}^{\playeridx} \ctrl{t}{i}. 
\end{equation}
% As noted by \cref{eq:lagrangian}, f
To formulate the Lagrangian for each agent $i$, we associate the initial condition constraint $0=\xstate{1}{i}-\initstate{i}$ with a Lagrange multiplier $\dynamicsInitLagrange{i}$, and at each step $t$, we associate a Lagrange multiplier $ \dynamicsLagrange{t}{i}$ with dynamics constraint \cref{eq:lq-dyn}.
For clarity, we denote $\DynamicsInitLagrange{}\equiv\{\dynamicsInitLagrange{i}\}_{i=1}^N$, $\dynamicsLagrange{t}{}\equiv\{\dynamicsLagrange{t}{i}\}_{i=1}^N$, and $\DynamicsLagrange{}\equiv\{\dynamicsLagrange{t}{}\}_{t=1}^{T-1}$.
For an LQ game, agent $i$'s Lagrangian can be defined as
\begin{equation}
    \begin{aligned}
        \LQlagrangian{i} (\StateTraj{}, \CtrlTraj{},\DynamicsInitLagrange{}, \DynamicsLagrange{};\gtparams{\playeridx}) \equiv & \costfn{i}(\StateTraj{}, \CtrlTraj{};\gtparams{\playeridx})+ \dynamicsInitLagrange{i} (\xstate{1}{i} - \initstate{i})+\\ 
        & \sum_{t=1}^{T-1} \dynamicsLagrange{t}{i}(\xstate{t+1}{i} - \LinearKKTA{}^i \xstate{t}{i} - \LinearKKTB{}^i \ctrl{t}{i}) .
    \end{aligned}
\end{equation}
Thus, we can derive LQ games' KKT conditions
% , as in \cref{eq: LQ KKT compact}, 
as follows:
\begin{equation}\label{eq: LQ KKT early}
    \forall i \in [N]\left\{\begin{aligned}
        & 0=\nabla_{( \StateTraj{}, \CtrlTraj{i}) } \LQlagrangian{i} (\StateTraj{}, \CtrlTraj{},\DynamicsInitLagrange{}, \DynamicsLagrange{};\gtparams{\playeridx})  ,  \\
        & 0=\xstate{t+1}{i}-\LinearKKTA{}^i\xstate{t}{i} - \LinearKKTB{}^i\ctrl{t}{i}, \forall t\in[T-1],  \\
        & 0=\xstate{1}{i}- \initstate{i}
    \end{aligned}\right.
\end{equation}
% We define the KKT condition of GLNE as
% \begin{equation}\label{eq: LQ KKT}
%     \begin{aligned}
%         & 0=\nabla_{x_{t+1}} \mathcal{L}^i(x, u, \lambda ), && \forall i\in[N] , \forall t \in [T]\\
%         & 0=\nabla_{u_t^i} \mathcal{L}^i(x, u, \lambda), && \forall i\in[N] , \forall t \in [T] \\
%         & 0=x_{t+1}-Ax_t - Bu_t, && \forall t\in[T]
%     \end{aligned}
% \end{equation}
To facilitate our analysis, we will compactly represent the KKT conditions as a linear equation, which requires defining the following matrices:
    \begin{equation*}
    \begin{aligned}
        Q(\gtparams{}) & \equiv \begin{bmatrix}Q(\gtparams{1}) \\ Q(\gtparams{2}) \\ \vdots \\ Q(\gtparams{N})  \end{bmatrix}, \bar{Q}(\gtparams{}) \equiv \begin{bmatrix}
            Q(\gtparams{}) & 0 & \cdots & 0 \\
            0 & Q(\gtparams{}) & \cdots & 0 \\
            \vdots & \vdots & \ddots & 0 \\
            0 & 0 & \cdots & Q(\gtparams{})
        \end{bmatrix},\\
            R&\equiv \begin{bmatrix}
                R^{1} & 0 & \cdots & 0 \\
                0 & R^{2} & \cdots & 0 \\
                \vdots & \vdots & \ddots & \vdots \\
                0 & 0 & \cdots & R^{N}
            \end{bmatrix}, \bar{R} := \begin{bmatrix}
            R & 0 & \cdots & 0 \\
            0 & {R} & \cdots & 0 \\
            \vdots & \vdots & \ddots & \vdots \\
            0 & 0 & \cdots & {R}
        \end{bmatrix}, \\
        \LinearKKTBarA{} & \equiv\left[\begin{array}{c c c c c  c}
            -\LinearKKTA{}^1 & I & 0 & \cdots & 0 & 0 \\
            0 & -\LinearKKTA{}^2 & I &  \cdots & 0 & 0 \\
            0 & 0 & -\LinearKKTA{}^3 & \cdots & 0  & 0\\
            \vdots & \vdots & \vdots & \ddots & \vdots & \vdots  \\
            0 & 0 & 0 & \cdots & -\LinearKKTA{}^N & I
        \end{array}\right],  \\
        \LinearKKTBarB{} & \equiv - \left[ \begin{array}{c c c c c c}
            \LinearKKTB{}^1 & 0 & 0 & \cdots & 0 & 0 \\
            0 & \LinearKKTB{}^2 & 0 & \cdots & 0 & 0 \\
            0 & 0 & \LinearKKTB{}^3 & \cdots & 0 & 0\\
            \vdots & \vdots  & \vdots & \ddots & \vdots & \vdots  \\
            0 & 0 & 0 & \cdots & \LinearKKTB{}^N & 0
        \end{array} \right],\\
        \LinearKKTBarC{}& \equiv \begin{bmatrix}
            I & 0 & 0 & \cdots & 0 & 0
        \end{bmatrix}.
    \end{aligned}
    \end{equation*}
    % \begin{equation*}
    %     \bar{B} := - \left[ \begin{array}{c c c c c c}
    %         B & 0 & 0 & \cdots & 0 & 0 \\
    %         0 & B & 0 & \cdots & 0 & 0 \\
    %         0 & 0 & B & \cdots & 0 & 0\\
    %         \vdots & \vdots  & \vdots & \ddots & \vdots & \vdots  \\
    %         0 & 0 & 0 & \cdots & B & 0
    %     \end{array} \right].
    % \end{equation*}
    In addition, we define
    \begin{equation}
        \LinearKKTBlkDiagA{} \equiv \begin{bmatrix}
            \LinearKKTA{}^1 & 0 & \cdots & 0  \\
            0 & \LinearKKTA{}^2 & \cdots & 0  \\
            \vdots & \vdots & \ddots & \vdots \\
            0 & 0  & \cdots & \LinearKKTA{}^N
        \end{bmatrix}, \LinearKKTBlkDiagB{} \equiv \begin{bmatrix}
            \LinearKKTB{}^1 & 0 & \cdots & 0 \\
            0 & \LinearKKTB{}^2 & \cdots & 0 \\
            \vdots & \vdots & \ddots & \vdots \\
            0 & 0 & \cdots & \LinearKKTB{}^N
        \end{bmatrix}
    \end{equation}
    and 
    \begin{align}
        \LinearKKTBigBlkDiagA{} &\equiv \left[\begin{array}{c c c c c  c}
            -\LinearKKTBlkDiagA{}^\top & 0 & 0 & \cdots & 0 \\
            I &  -\LinearKKTBlkDiagA{}^\top &  &  \cdots & 0  \\
            0 & I & -\LinearKKTBlkDiagA{}^\top & \cdots & 0  \\
            \vdots & \vdots & \vdots & \ddots & \vdots   \\
            0 & 0 & 0 & \cdots & I
        \end{array}\right],\LinearKKTBigBlkDiagC{} \equiv \begin{bmatrix}
            I \\ 0 \\ 0 \\ \vdots \\ 0
        \end{bmatrix}, \\
        \LinearKKTBigBlkDiagB{} &\equiv - \left[ \begin{array}{c c c c c c}
            \LinearKKTBlkDiagB{}^\top & 0 & 0 & \cdots & 0  \\
            0 & \LinearKKTBlkDiagB{}^\top & 0 & \cdots & 0  \\
            0 & 0 & \LinearKKTBlkDiagB{}^\top & \cdots & 0 \\
            \vdots & \vdots  & \vdots & \ddots & \vdots  \\
            0 & 0 & 0 & \cdots  & \LinearKKTBlkDiagB{}^\top \\
            0 & 0 & 0 & \cdots  & 0
        \end{array} \right].
    \end{align}
    % We also define $\LinearKKTBigBlkDiagC{} \equiv \begin{bmatrix}
    %         I & 0
    %     \end{bmatrix}$ with appropriate dimensions to fit in the matrix $M(\GtParams{})$. 
    % \begin{equation}    
        % \begin{bmatrix}
        %     \bar{C} & 0 & \cdots & 0 \\
        %     0 & \bar{C} & \cdots & 0 \\
        %     \vdots & \vdots & \ddots & \vdots \\
        %     0 & 0 & \cdots & \bar{C}
        % \end{bmatrix}
    % \end{equation}
    We can compactly write the KKT conditions
    \begin{equation}
        \underbrace{\left[\begin{array}{c c  c c }
            \bar{Q}({\GtParams{}}) & 0 & \LinearKKTBigBlkDiagA{} & \LinearKKTBigBlkDiagC{} \\
            0 & \bar{R} & \LinearKKTBigBlkDiagB{} & 0 \\  
            \LinearKKTBarA{} & \LinearKKTBarB{} & 0 & 0 \\
            \LinearKKTBarC{} & 0 & 0 & 0
        \end{array}\right]}_{M(\GtParams{})} \underbrace{\begin{bmatrix}
            \StateTraj{} \\ \CtrlTraj{} \\ \DynamicsLagrange{} \\ \DynamicsInitLagrange{}
        \end{bmatrix}}_{\LQtotalvariable} + \underbrace{\begin{bmatrix}
            0 \\ 0\\ 0 \\ -I
        \end{bmatrix}}_{S} \initstate{} = 0
    \end{equation} 
    where the first two rows correspond to the first row of \eqref{eq: LQ KKT early}, and each remaining row corresponds to one of the subsequent rows in \eqref{eq: LQ KKT early}. 
    % where each row corresponds to a row in \eqref{eq: LQ KKT}

%     % Let $n_z=T(n_x + n_u + n_\lambda)$ be the total dimension of the all variables we considered in the KKT conditions. 
%     Throughout this section, we assume that there is a unique LGNE and that $M(\GtParams{})$ is invertible. Precise conditions ensuring the uniqueness of the solution and invertibility of $M(\GtParams{})$ can be found in \cite[Section 6.2]{bacsar1998dynamic}. %The matrix $M(Q)$ would be invertible under the condition that all state cost matrices $Q^i$ and control matrices $R^i$ are positive definite. 
    
%     % To compute the player $i$'s LGNE control under its belief $\{\hat{Q}^{ij}\}_{j=1}^N$, we have
%     % \begin{equation}
%     %     M(\hat{Q}^i) z^i + K x_0 = 0
%     % \end{equation}
%     % where $\hat{Q}^i = \{\hat{Q}^{ij}\}_{j=1}^N$ and $z^i\in\mathbb{R}^{n_z}$.
%     % Similarly, we compute the LGNE solution $z'$ under level-1 inference as 
%     % \begin{equation}
%     %     M(Q') z' + K x_0 = 0
%     % \end{equation}
%     % In this subsection, we simplify the inverse-game loss function as the error between the predicted OLNE control under the ground truth second-order belief and the one predicted using level-1 belief. 
%     % \begin{equation}
%     %     L(Q) = \sum_{t=0}^T \|\pi(x_0|Q^*) - \pi(x_0|Q)\|_2^2
%     % \end{equation}
    
% % We will first characterize the solution space of level-1 and level-2 OLNE inference. 
% % \begin{proposition} \label{prop:nonconvex}
% %     Minimizing the level-2 inverse game loss function in the level-2 belief space is a non-convex problem.
% % \end{proposition}
\begin{proof}[Proof of Proposition~\ref{prop:nonconvex main text}]
    We approach the proof by presenting a counterexample complying with the conditions in Proposition~\ref{prop:nonconvex main text}'s statement, showing that the convex combination of two parameters $\GtParamsguess{}_1$ and $\GtParamsguess{}_2$ does not necessarily lead to a lower level-2 inverse game loss. 
    
    Suppose $A = I_2$, $B^i = I_2$, $R^i = I_2$, $\forall i\in\{1,2\}$. Let $T=2$, and $x_0=[1,-1]^\top$.  Consider the ground truth level-2 parameters $ \GtParamsguess{1*}\equiv \{\params{1*}, \params{1,2*}\}$ and $ \GtParamsguess{2*}\equiv \{ \params{2*}, \params{2,1*} \} $ are such that: 
    \begin{equation}
        \begin{aligned}
            & Q( \params{1*} )  =   \begin{bmatrix}
        0.1 & 0 \\ 0 & 0.1
    \end{bmatrix}, && Q(\params{1,2*})  =\begin{bmatrix}
        1 & 0 \\ 0 & 1
    \end{bmatrix}\\
    & Q(\params{2,1*})  = \begin{bmatrix}
        0.1 & 0.1 \\ 0.1 & 0.1
    \end{bmatrix},  &&Q(\params{2*}) =\begin{bmatrix}
        1 & 1 \\ 1 & 1
    \end{bmatrix}
        \end{aligned}
    \end{equation}
     Let one set of solutions, $\GtParamsguess{}_1=\{\GtParamsguess{1}_1, \GtParamsguess{2}_1\}$, $ \GtParamsguess{1}_1 =\{ \params{1}_1,\params{1,2}_1  \} $ and $\GtParamsguess{2}_1 = \{\params{2}_1, \params{2,1}_1\}$ be such that
    \begin{equation*}
    \begin{aligned}
        &Q(\params{1}_1)  = \begin{bmatrix}
            1 & 0.8 \\ 0.8 & 1
        \end{bmatrix}, && Q(\params{1,2}_1)  = Q(\params{1,2*}) , \\
        &Q(\params{2,1}_1) = Q(\params{2,1*}), && Q(\params{2}_1) = Q(\params{2*})
    \end{aligned}
    \end{equation*}
    and another set of solutions, $\GtParamsguess{}_2 =\{\GtParamsguess{1}_2, \GtParamsguess{2}_2\}$, $\GtParamsguess{1}_2 = \{ \params{1}_2, \params{1,2}_2\}$, and $ \GtParamsguess{2}_2 = \{ \params{2}_2, \params{2,1}_2 \} $ be such that
    \begin{equation*}
    \begin{aligned}
        & Q(\params{1}_2)) = \begin{bmatrix}
            1 & -0.8 \\ -0.8 & 1
        \end{bmatrix}, && Q(\params{1,2}_2) = Q(\params{1,2*}), \\ & Q(\params{2,1}_2) = Q(\params{2,1*}), &&Q(\params{2}_2)=Q(\params{2*})
    \end{aligned}
    \end{equation*}
    We have $\loss( \GtParamsguess{}_1 )=0.0612$, $\loss ( \GtParamsguess{}_2 ) = 0.6025$, and $\loss (\frac{1}{2}( \GtParamsguess{}_1 +  \GtParamsguess{}_2 ))=0.4041$. However, $\frac{1}{2} (\loss ( \GtParamsguess{}_1 ) + \loss ( \GtParamsguess{}_2 )) < \loss(\frac{1}{2}( \GtParamsguess{}_1 +  \GtParamsguess{}_2 )) $, and this shows the function $\loss ( \GtParamsguess{} )$ is non-convex with respect to $ \GtParamsguess{}$. 
\end{proof}

\begin{proof}[Proof of Proposition~\ref{prop:bounds}]
    We first prove the upper bound. Denote by $\check{\LQtotalvariable}$ the solution to the KKT conditions $M(\constructedlevelonesol) \check{\LQtotalvariable} + S \initstate{}=0$. Let $\delta \LQtotalvariable^j := \check{\LQtotalvariable} - \LQtotalvariable^j$ and $\Delta^j := M(\constructedlevelonesol  ) - M( \secondgtparams{j*} )$. Since $\constructedlevelonesol$ is a constructed solution to the level-1 inference, we have %Define one candidate level-1 inference solution $Q':=[Q^{ii*}]_{i=1}^N$. 
    \begin{equation}
    \begin{aligned}
        \loss(\LevelOneGtParamsguess^*) & \le \loss( \constructedlevelonesol ) \le \frac{1}{2} \sum_{j=1}^N\left\| \delta \LQtotalvariable^j \right\|_2^2
    \end{aligned}
    \end{equation}
    Note that 
    \begin{equation}
    \begin{aligned}
        \delta \LQtotalvariable^j & = M( \constructedlevelonesol )^{-1} S \initstate{} - \LQtotalvariable^j\\
        & = M( \constructedlevelonesol)^{-1} M(\secondgtparams{j*}) \LQtotalvariable^j - \LQtotalvariable^j\\
        & = M(\constructedlevelonesol)^{-1} ( M(\constructedlevelonesol) + M(\secondgtparams{j*}) - M(\constructedlevelonesol)) \LQtotalvariable^j - \LQtotalvariable^j \\
        & = -M(\constructedlevelonesol)^{-1} \Delta^j \LQtotalvariable^j
    \end{aligned}
    \end{equation}
    and we have 
    \begin{equation}
    \begin{aligned}
        \loss( \constructedlevelonesol) \le& \frac{1}{2}\sum_{j=1}^N \|M(\constructedlevelonesol)^{-1} \Delta^j \LQtotalvariable^j\|_2^2 \\
        \le & \frac{1}{2}\sum_{j=1}^N \left(\frac{\sigma_{\max}(M( \secondgtparams{j*})-M(\constructedlevelonesol    ))}{\sigma_{\min}(M(\constructedlevelonesol))} \|\LQtotalvariable^j(\secondgtparams{j*})\|_2\right)^2
    \end{aligned}
    \end{equation}
    
    Next, we show the lower bound. We observe that the level-1 inference loss function can be simplified as 
    \begin{equation}
        \loss(\LevelOneGtParamsguess) =\frac{1}{2} \sum_{j=1}^N \left\| E^j \left(M( \secondgtparams{j*} )^{-1} S \initstate{} - M( \LevelOneGtParamsguess )^{-1} S \initstate{}  \right) \right\|_2^2
    \end{equation}
    Assuming $Y = M(\LevelOneGtParamsguess)^{-1}$, and setting the gradient of the above loss over $Y$ to be zero, we have 
    \begin{equation}
        0 = \sum_{j=1}^N M( \secondgtparams{j*} )^{-1} - N Y^*
    \end{equation}
    Solving this equation with respect to $Y$, we can obtain an optimal $Y^*$
    \begin{equation}
        Y^* :=\frac{1}{N}  \sum_{j=1}^N M( \secondgtparams{j*} )^{-1} 
    \end{equation}
    and we have 
    \begin{equation}
        \sum_{j=1}^N \left( M( \secondgtparams{j*})^{-1} S \initstate{} - Y^* S \initstate{} \right) = 0
    \end{equation}
    Observe $\check{M} = (Y^*)^{-1}$, we have
    \begin{equation*}
    \begin{aligned}
        \loss( \LevelOneGtParamsguess) & \ge\frac{1}{2} \sum_{j=1}^N \left\| E^j \left( \check{M}^{-1} S \initstate{} -  M( \secondgtparams{j*})^{-1} S \initstate{} \right) \right\|_2^2 \\
        & = \frac{1}{2}\sum_{j=1}^N \left\| E^j \left( \check{M}^{-1} ( \check{M} + M( \secondgtparams{j*}) - \check{M} )\LQtotalvariable - \LQtotalvariable^j \right) \right\|_2^2 \\
        & = \frac{1}{2}\sum_{j=1}^N \left\| E^j \check{M}^{-1}\left( M( \secondgtparams{j*}) - \check{M} \right) \LQtotalvariable \right\|_2^2\\
        & \ge \frac{1}{2}\sum_{j=1}^N \left(\frac{\|E^j (M( \secondgtparams{j*}) - \check{M})\LQtotalvariable(\secondgtparams{j*})\|_2}{\sigma_{\max} (\check{M})}\right)^2
    \end{aligned}
    \end{equation*}
    % Since $M(Q^{j*})z^j = K x_0$, and $E^j Kx_0 = 0$, we have 
    % \begin{equation}
    %     L_1(\hat{Q}) \ge \frac{1}{2}\sum_{j=1}^N \frac{\|E^j \check{M} z^j\|_2^2}{(\sigma_{\max} (\check{M}))^2}
    % \end{equation}
\end{proof}

%% file: include/sec7_appendix.tex
\section{Supplement B: LQ Game Experimental Details}
\label{appendix:lq-game-implementation-details}

We set the random seed to 42 for all experiments involving randomness. All experiments are run on a Nov. 2024 Apple MacBook Pro with 24GB of RAM running MacOS Sequoia 15.5. All code runs in Julia v1.11.5.
% with the following dependencies (taken from the \texttt{Project.toml} file included with the code.

% Generate the forward game based on KKT conditions.

% How is GT generated + $Q^*$ matrices?
% How do we define heterogeneity/homogeneity? - done after (18)
% Gradient descent parameters + initial solution

\subsection*{Environment Setup}
We consider a class of two-player finite-horizon linear-quadratic (LQ) games over a discrete-time horizon $\horizon = 3$ steps, as described by \cref{eq:lq-cost,eq:lq-dyn}.
Let one such game have two-dimensional ($\numstates{} = \numctrls{\playeridx{}} = 2$) state and controls $\xstate{t}{}, \ctrl{t}{\playeridx{}} \in \mathbb{R}^2$,
with control cost and dynamics matrices $R^{\playeridx{}} = A =  B^{\playeridx{}} = I_{2\times 2} $.
The initial state $\initstate{}{} = [10 ~ 10]^{\T}$ is fixed for each experiment.
As with our theoretical derivation, we assume no noise in this game.

% Q^{\playeridx{}} =

\subsubsection{Ground Truth Fictitious Gameplay}
We define game parameters 
\begin{equation}
    \secondgtparams{i} = [\underbrace{\gtparams{(i,1)}~\gtparams{(i,2)}~\gtparams{(i,3)}}_{\gtparams{i, 1}}~\underbrace{\gtparams{(i,4)}~\gtparams{(i,5)}~\gtparams{(i,6)}}_{\gtparams{i, 2}}] \in \mathbb{R}^6,
\end{equation} where parameterized state costs are written as
\begin{align}
    Q(\gtparams{i, 1}) &= \left[ \begin{array}{cc}
        \gtparams{(i, 1)} & \gtparams{(i,3)} \\
        \gtparams{(i,3)} & \gtparams{(i,2)}
    \end{array} \right],\\
    Q(\gtparams{i, 2}) &= \left[ \begin{array}{cc}
        \gtparams{(i,4)} & \gtparams{(i,6)} \\
        \gtparams{(i,6)} & \gtparams{(i,5)}
    \end{array} \right].
\end{align}
Given a hyperparameter $s \in \mathbb{R}$ which scales the level of heterogeneity, we define the following parameters for each agent ($s$ is underlined for clarity):
\begin{align}
    \secondgtparams{1} &= [1, ~1,~-\!\!1,~10,~10\underline{s},~1],\\
    \secondgtparams{2} &= [10,~10,~1,~1,~1,~-1].
\end{align}
We measure heterogeneity using the equation
\begin{equation}
    \sum_{i=1}^N\|M( \secondgtparams{i*} ) - \check{M}\|_2,
\end{equation}
as described in the main text.
This equation measures the difference between each objective estimate matrix and the average objective estimate matrix across all agents.

% \subsection*{Ground Truth Gameplay Generation}

% To generate observed behavior, we solve the Karush-Kuhn-Tucker (KKT) system corresponding to each agent’s optimization problem. The full system is defined in \cref{eq:lq-kkt-conditions}. 
% Solving this system (using \cref{eq: LQ kkt equation}) yields the equilibrium strategies each agent plays given its internal beliefs \( \Params{}^\playeridx{} \). 
% These strategies are used to simulate trajectories in forward play.
% Ground truth gameplay is generated using the state cost matrices
% \todo{Add matrices}.

\subsection*{Inverse Game Solver and Optimization}

We minimize the inverse loss function \( \mathcal{L}(\hat{\Theta}) \) from \cref{eq:lq-inverse-game-loss} using gradient descent.
Gradients are computed via implicit differentiation through the KKT system using \texttt{ForwardDiff.jl} \cite{RevelsLubinPapamarkou2016}.
To perform offline level-2 inference, we solve the KKT conditions and apply a line search strategy, as described by \cite[\S 3]{wright1999numerical} at each step to select the step size.
We initial the guess for $Q^1_0$ and $Q^2_0$ to be identity $I_{2}$.
% \todo{Update notation}
At every iteration of gradient descent, we project each $Q^{i,j}$ matrix to be positive semi-definite
% All cost matrices \( Q^i \) are constrained to be positive semi-definite (PSD) 
to ensure convexity of each player’s objective. 
We only consider cases where the KKT matrix \( M(\Theta) \) is invertible, guaranteeing a unique LGNE.

% using \cref{alg:level-2 inference}
In the process of inferring the level-2 parameters, we utilize optimization hyperparameters
\begin{align}
K &= 5000 && \text{(maximum iterations),} \\
\beta &= 0.5 && \text{(step size decay parameter), and} \\
\alpha &= 10 && \text{(initial learning rate).}
\end{align}

% \subsection*{Observation Model}

% We assume full observation of each agent’s action trajectory:
% \[
% \obs{t}{\playeridx{}} = \ctrl{t}{\playeridx{}} + \epsilon^i_t, \quad \epsilon^i_t \sim \mathcal{N}(0, \sigma^2 I),
% \]
% with fixed noise level \( \sigma \). The loss function compares observed and predicted actions under this model.

\section{Supplement C: Lane Change Experimental Details}
\label{appendix:lc-implementation-details}

\subsection{Environment Setup}
We consider a two-lane road with lane width $l_w = \SI{2.0}{\meter}$.
The road spans from $0$ to $\SI{4.0}{\meter}$ in the lateral direction.

We model each %agent
agent's state as evolving in discrete time with sampling interval $\Delta t$ according to planar double-integrator dynamics $\dynamics{\playeridx{}}$, i.e. % in the longitudinal and lateral coordinates, 
\begin{equation}
\label{eq:exp-planar-di-dynamics}
\xstate{t+1}{\playeridx}
= \left[ \begin{array}{c}
    p^{\playeridx}_{t+1} \\
    v^{\playeridx}_{t+1} \\
\end{array} \right]
= \left[ \begin{array}{c}
    p^{\playeridx}_{t} + \Delta t \cdot v^{\playeridx}_{t} \\
    v^{\playeridx}_{t} + \Delta t \cdot a^{\playeridx}_{t}  \\
\end{array} \right] \in \mathbb{R}^4.
\end{equation}
At time $t$, $p^{\playeridx}_{t} = [p^{\playeridx}_{\lat{},t} ~ p^{\playeridx}_{\lon{},t}]^\T \in \mathbb{R}^2$ is agent $\playeridx$'s planar position, 
$v^{\playeridx}_{t} = [v^{\playeridx}_{\lat{},t} ~ v^{\playeridx}_{\lon{},t}]^\T \in \mathbb{R}^2$ is agent $\playeridx$'s planar velocity, and 
$\ctrl{t}{\playeridx} = [a^{\playeridx}_{\lat{},t} ~ a^{\playeridx}_{\lon{},t}]^\T$ is the commanded planar acceleration.
%applied to the point.
The subscripts \lat{} and \lon{} signify the %latitudinal
lateral (lane offset) and longitudinal (along-lane) directions, respectively.

Each agent's objective $\costfn{\playeridx}$ is composed of stage costs
\begin{equation}
    \label{eq:lc-objective}
    \stagecostfn{t}{\playeridx}(\xstate{t}{}, \ctrl{t}{}; \gtparams{\playeridx, \playeridx}) \equiv 
    w^{\playeridx{}}_1 \underbrace{(p^{\playeridx{}}_{\lat{},t} - \gtparams{\playeridx{},\playeridx{}})^2}_{(\ref{eq:lc-objective}a)} 
    + w^{\playeridx{}}_2 \underbrace{\| v^{\playeridx}_t - v_d \|_2^2}_{(\ref{eq:lc-objective}b)} 
    + w^{\playeridx{}}_3 \underbrace{\| \ctrl{t}{\playeridx{}} \|_2^2}_{(\ref{eq:lc-objective}c)},
\end{equation}
% \david{no vec notation} \david{also shouldn't it be $p_{y, t}^i$ (y and i)?}
where predetermined weights $w^i_j > 0, \forall i, j$, set %component costs' 
relative priorities between competing incentives: tracking a desired lane offset parameter (\ref{eq:lc-objective}a), maintaining a desired velocity (\ref{eq:lc-objective}b), and minimizing control costs (\ref{eq:lc-objective}c).

Each agent also considers a collision-avoidance constraint, which couples its decision with that of the other agent and is parameterized by the safety buffer $\delta \in \mathbb{R}$:
\begin{equation}
    \label{eq:collision-avoidance-constraint}
    \| p^{1}_t - p^{2}_t \|^2_2 \geq \delta^2, \qquad \qquad\forall t \in [\horizon].
\end{equation}
Finally, we enforce upper and lower bounds on the state and controls to ensure each vehicle remains on the road and produces realistic control inputs.
We provide additional details and hyperparameter values below.

% \subsubsection{Road Configuration} We consider a two-lane road with lane width $l_w = \SI{2.0}{\meter}$.
% The road spans from $0$ to $\SI{4.0}{\meter}$ in the lateral direction.
% accommodating two lanes.

\subsubsection{Vehicle Dynamics} Each agent $i$ evolves according to planar double-integrator dynamics as specified in \cref{eq:exp-planar-di-dynamics}.
The discrete-time sampling interval is $\Delta t = \SI{0.1}{\second}$.

\subsubsection{Initial Conditions} Agent 1 starts at position 
\[p_{init}^1 = [\SI{1.0}{\meter}, \SI{1.0}{\meter}]^\top\] 
with initial velocity
\[v_{init}^1 = [\SI{0.0}{\meter\per\second}, \SI{1.0}{\meter}]^\top.\]
Agent 2 starts at position 
\[p_{init}^2 = [\SI{3.2}{\meter},\SI{0.9}{\meter}]^\top\]
with initial velocity 
\[v_{init}^2 = [\SI{0.0}{\meter\per\second},\SI{1.0}{\meter\per\second}]^\top.\]
Both agents have desired velocity 
\[v_d = [\SI{0}{\meter\per\second},\SI{2.0}{\meter\per\second}]^\top.\]

\subsection{Objective Function Parameters}

The stage cost $\ell_t^i$ defined in \cref{eq:lc-objective} uses the weight parameters
\begin{align}
w_1^i &= 1.0 && \text{(lane tracking weight),} \label{eq:lc-weight-1} \\
w_2^i &= 0.5 && \text{(velocity tracking weight), and} \label{eq:lc-weight-2} \\
w_3^i &= 0.1 && \text{(control effort weight).} \label{eq:lc-weight-3}
\end{align}

\subsubsection{True Parameters} Both agents desire to track the middle of the top lane: $\gtparams{1*} = \gtparams{2*} = \gtparams{(1,1)* = \gtparams{(2,2)*}} = \nicefrac{l_w}{2} = \SI{1.0}{\meter}$.
For the parameter sweep in \cref{fig:forward-lc-sweep}, we vary agents' parameters of each other's desired lane offsets:
\begin{align}
\gtparams{(1,2)*}, \gtparams{(2,1)*} \in [\SI{0.5}{\meter} , \SI{4.0}{\meter}] \times [\SI{0.5}{\meter} , \SI{4.0}{\meter}].
\end{align}
We generate 64 simulation instances with true parameters from the provided ranges.

\subsection{Constraints and Bounds}
In our experiments, we choose appropriate values for
% mass ($\mass = \SI{1}{\kilogram}$), 
velocity and acceleration parameters.
% in order to simplify the presentation of our work.
Adjusting these parameters to standard ranges used in driving should result in scaled changes in outputs, and we expect our method generalizes in a straightforward manner. 

\subsubsection{State Constraints} Vehicle positions are bounded to stay on the road,
\begin{align}
p_{lat,t}^i &\in [\SI{0.0}{\meter}, \SI{4.0}{\meter}], \\
p_{lon,t}^i &\in [\SI{0.0}{\meter}, \SI{50.0}{\meter}],
\end{align}
and velocities are bounded to maintain certain speed limits:
\begin{align}
v_{\lat,t}^i &\in [\SI{-10.0}{\meter\per\second}, \SI{10.0}{\meter\per\second}] \\
v_{\lon,t}^i &\in [\SI{0.0}{\meter\per\second}, \SI{10.0}{\meter\per\second}]
\end{align}

\subsubsection{Control Constraints} Acceleration inputs are limited to:
\begin{equation}
a_{\lat,t}^i, a_{\lon,t}^i \in [\SI{-5}{\meter\per\second\squared}, \SI{3}{\meter\per\second\squared}].
\end{equation}

\subsubsection{Safety Constraint} 
The safety buffer constraint \cref{eq:collision-avoidance-constraint} enforces a minimum distance of $\delta = \SI{2.0}{\meter}$ between vehicles at all times.

\subsection{Generating Ground Truth Receding Horizon Play}
To generate ground truth fictitious gameplay trajectories for a set of parameters $\gtparams{1,1}, \gtparams{1,2}, \gtparams{2,1}, \gtparams{2,2}$, each agent solves a hypothesized Nash game $\Gamma(\secondgtparams{i})$ in a receding horizon manner over a horizon of $150$ steps (\SI{15.0}{\second}).
Each game generates a plan over a prediction horizon of $\horizon = 15$ steps (\SI{1.5}{\second}), and re-planning occurs at every third time step.
These parameters provide sufficient planning horizon and simulation length to observe the complex behavior described in \cref{fig:forward-lc-sweep}, including initial deadlock and successful lane change maneuvers.

\subsection{Observation Model and Noise}

\subsubsection{Measurement Setup} The third-party observer receives noisy position measurements
\begin{equation}
y_t^i \sim \mathcal{N}(p_t^i, c^2 I_2),
\end{equation}
where standard deviation hyperparameter $c = 0.1 \SI{}{\meter}$.
Position observations are available at every time step, and observation model
\begin{equation}
    G(\xstate{t}{i}) = \left[\begin{array}{cccc}
         1 & 0 & 0 & 0  \\
         0 & 1 & 0 & 0
    \end{array}\right] \xstate{t}{i}. 
\end{equation}

\subsection{Solving \texorpdfstring{$\Gamma(\secondgtparams{i})$}~~  within the Inverse Game}
In the process of inferring the level-2 parameters,
% using \cref{alg:level-2 inference}, 
we utilize optimization hyperparameters
\begin{align}
K &= 40 && \text{(maximum iterations),} \\
\epsilon &= 0.1 && \text{(convergence threshold), and} \\
\alpha &= 0.1 && \text{(initial learning rate).}
\end{align}

To generate ground truth fictitious gameplay trajectories for a set of parameters $\gtparams{(1,1)*}, \gtparams{(1,2)*}, \gtparams{(2,1)*}, \gtparams{(2,2)*}$, each agent solves a hypothesized Nash game $\Gamma(\secondgtparams{\playeridx*})$ in a receding horizon manner over a horizon of $150$ steps (\SI{15.0}{\second}).
Each game generates a plan over a prediction horizon of $\horizon = 15$ steps (\SI{1.5}{\second}), and re-planning occurs at every third time step.
These parameters provide sufficient planning horizon and simulation length to observe the complex behavior described in \cref{fig:forward-lc-sweep}, including initial deadlock and successful lane change maneuvers.

% \subsubsection{Horizon Parameters} Each agent solves its hypothesized Nash game $\Gamma(\Params{i})$ over a prediction horizon using $\horizon = 15$ steps (\SI{1.5}{\second}) of data. \hmzh{Re-planning occurs at every third time step.}

% \subsubsection{Simulation Duration} The total simulation runs for \SI{15.0}{\second} (150 time steps), providing sufficient data for both initial deadlock and subsequent lane change maneuvers.

\subsection{Algorithm Configuration}

\begin{algorithm}[t]
\SetAlgoLined
\caption{Level-2 Inverse Game Solver}
\label{alg:level-2 inference}
\KwData{Observation data $\ObsTraj{}$, initial parameter guess $\GtParamsguess{}$, maximum iteration number $K$, convergence threshold $\epsilon > 0$}
\KwResult{Estimated parameter $\GtParamsguess{}$}
\For{$k=0,1,\dots,K$}{
    Obtain $\nabla \loss(\GtParamsguess{})$ via \eqref{eq:loss gradient} and a differentiable MCP solver \\
    Gradient descent update with stepsize $\alpha > 0$: $\GtParamsguess{} \gets \GtParamsguess{} - \alpha \nabla \loss (\GtParamsguess{})$ \\
    % Record solution: $\GtParamsguess{}\gets \GtParamsguess{}{'}$ \\
    Terminate if $\loss(\GtParamsguess{})\le \epsilon$
}
\Return{$\GtParamsguess{}$}
\end{algorithm}

\subsubsection{Level-2 Inverse Game Solver}
At each iteration of the online inference, we apply our offline inference method using optimization parameters
\begin{align}
K &= 40 && \text{(maximum iterations),} \\
\epsilon &= 0.1 && \text{(convergence threshold), and} \\
\alpha &= 0.1 && \text{(initial learning rate)}.
\end{align}

% The learning rate follows an adaptive schedule: $\alpha_k = 0.05 / (1 + 0.01k)$ to ensure convergence stability.

\subsubsection{Online Inference} 
We solve level-2 inverse problem in real-time on a scenario with dynamically changing parameters.
% \cref{alg:level-2 inference}
Our proposed method assumes a static parameter value over the horizon it analyzes.
We introduce a receding horizon algorithm which analyzes $\Tsim = 15$ steps (\SI{1.5}{\second}) of data at a time, so we infer parameters from time $t = 15$ until $t = 150$, the number of steps in the simulation.
To run level-2 inference online, we call our static method in each iteration and then update the parameters for the receding horizon method according to the results of the static inference.
% Within each loop, the algorithm sets up the arguments for the call to \cref{alg:level-2 inference}, makes the call, and then prepares for the next iteration using the results.

\subsubsection{Parameter Initialization}
The initial parameter estimates are set to $\hat{\theta}^{i,j} = \SI{2.0}{\meter}$ for all $i,j$ because it indicates maximum uncertainty (center of the road).
Future inverse games are solved provide the previous parameter estimate.

% \david{Increase reference font size.}

\subsection{MCP Solver Configuration}
We use the PATH solver \cite{dirkse1995path} for computing Nash equilibria in each hypothesized game $\Gamma(\hat{\Theta}^i)$ employs
\begin{align}
\text{Convergence tolerance} &= 10^{-6}, \\
\text{Maximum iterations} &= 10^5, \text{ and } \\
\text{Complementarity tolerance} &= 10^{-2}.
\end{align}
Calls to the PATH solver require an initial guess for all primal and dual variables.
These are chosen given an initial state, which is propagated forward across the provided horizon with a zero-control trajectory for each agent.
All dual variables are initialized to 0.

% \subsubsection{Numerical Stability} To ensure robust optimization, we apply logarithmic barrier regularization with weight $\mu = 10^{-6}$ for parameter bound constraints and add $10^{-8} \cdot I$ to the Hessian for positive definiteness.

% \subsection{\todo{Prediction Costs Over Time for Lane Change Scenarios}}
% \label{appendix:lc-prediction-costs}

% \input{figs/mismatched_prediction_costs_LCs}

% \subsection{\todo{Additional Plots for inD Inference}}
% \label{appendix:ind-inference-plots}
% \todo{Switch out these plots if some are interesting.}

% \input{figs/inD_figure_extra}

% \include{include/lq_source_code}

%% file: include/notation_tables.tex
\begin{table*}[t]
\caption{Notation Reference I: Game Setup, States, Controls, and Observations}
\small
% \begin{tabular}{@{}p{1.3cm}p{3.2cm}p{9cm}@{}}
\begin{tabularx}{\textwidth}{@{}p{2.5cm}p{4cm}X@{}}
\toprule
\textbf{Category} & \textbf{Symbol} & \textbf{Description} \\
\midrule
\multirow{5}{*}{\textbf{Game Setup}}
& $\numplayers$ & Number of agents in the game (positive integer) \\
& $\horizon$ & Length of the finite time horizon (positive integer) \\
& $[d]$ & Set $\{1,2,\ldots,d\}$ for positive integer $d$ \\
& $t \in [\horizon]$ & Discrete time step index \\
& $\playeridx, \secondidx \in [\numplayers]$ & Primary and secondary agent indices \\
\midrule
\multirow{9}{*}{\textbf{States}}
& $\numstates{\playeridx}$ & State dimension of agent $\playeridx$ (positive integer) \\
& $\xstate{t}{\playeridx} \in \mathbb{R}^{\numstates{\playeridx}}$ & State of agent $\playeridx$ at time $t$ \\
& $\numstates{} = \sum_{\playeridx=1}^{\numplayers} \numstates{\playeridx}$ & Total joint state dimension \\
& $\xstate{t}{} \equiv \{\xstate{t}{1}, \xstate{t}{2}, \ldots, \xstate{t}{\numplayers}\}$ & Joint state vector of all agents at time $t$ \\
& $\xstate{t}{-\playeridx} \equiv \xstate{t}{} \setminus \{\xstate{t}{\playeridx}\}$ & Joint state excluding agent $\playeridx$ at time $t$ \\
& $\StateTraj{\playeridx} \equiv \{\xstate{1}{\playeridx}, \xstate{2}{\playeridx}, \ldots, \xstate{\horizon}{\playeridx}\}$ & State trajectory of agent $\playeridx$ over horizon $\horizon$ \\
& $\StateTraj{} \equiv \{\StateTraj{1},\StateTraj{2},\ldots,\StateTraj{\numplayers}\}$ & Joint state trajectory of all agents over horizon \\
& $\StateTraj{-\playeridx} \equiv \StateTraj{} \setminus \{\StateTraj{\playeridx}\}$ & Joint state trajectory excluding agent $\playeridx$ \\
& $\initstate{}$ & Joint initial state \\
\midrule
\multirow{7}{*}{\textbf{Controls}}
& $\numctrls{\playeridx}$ & Control dimension for agent $\playeridx$ (positive integer) \\
& $\numctrls{} = \sum_{\playeridx=1}^{\numplayers} \numctrls{\playeridx}$ & Total joint control dimension \\
& $\ctrl{t}{\playeridx} \in \mathbb{R}^{\numctrls{\playeridx}}$ & Control of agent $\playeridx$ at time $t$ \\
& $\ctrl{t}{} = \{\ctrl{t}{1}, \ctrl{t}{2}, \ldots, \ctrl{t}{\numplayers}\}$ & Joint control vector at time $t$ (analogous to $\xstate{t}{}$) \\
& $\CtrlTraj{\playeridx} \equiv \{\ctrl{1}{\playeridx}, \ctrl{2}{\playeridx}, \ldots, \ctrl{\horizon}{\playeridx}\}$ & Control trajectory of agent $\playeridx$ over horizon \\
& $\CtrlTraj{} \equiv \{\CtrlTraj{1},\CtrlTraj{2},\ldots,\CtrlTraj{\numplayers}\}$ & Joint control trajectory of all agents over horizon \\
& $\dynamics{\playeridx}$ & Dynamics function \\
% \cref{eq:dynamics} \\
\midrule
\multirow{7}{*}{\parbox{1.3cm}{\textbf{Obser-vations}}}
& $\numobs{\playeridx}$ & Observation dimension for agent $\playeridx$ (positive integer) \\
& $\obs{t}{\playeridx} \in \mathbb{R}^{\numobs{\playeridx}} \cup \{\emptyset\}$ & Observation of agent $\playeridx$ at time $t$ (possibly missing) \\
& $\obs{t}{} \equiv \{\obs{t}{1}, \obs{t}{2}, \ldots, \obs{t}{\numplayers}\}$ & Joint observation vector at time $t$ \\
& $\ObsTraj{\playeridx} \equiv \{\obs{1}{\playeridx}, \obs{2}{\playeridx}, \ldots, \obs{\horizon}{\playeridx}\}$ & Observation trajectory of agent $\playeridx$ over time \\
& $\ObsTraj{} \equiv \{\ObsTraj{1},\ObsTraj{2},\ldots,\ObsTraj{\numplayers}\}$ & Joint observation trajectory (all agents, all times) \\
& $\ObservationFn(\StateTraj{}, \CtrlTraj{})$ & Observation function mapping state and control trajectories to observations \\
& $p(\ObsTraj{}|\StateTraj{}, \CtrlTraj{})$ & Likelihood: conditional probability of observations given trajectories \\
\bottomrule
\end{tabularx}
\label{tab:notation_game_setup}
\end{table*}

\begin{table*}[t]
\caption{Notation Reference II: Parameters, Costs, Constraints, and Nash Games}
\small
% \begin{tabular}{@{}p{1.3cm}p{3.2cm}p{9cm}@{}}
\begin{tabularx}{\textwidth}{@{}p{2.5cm}p{4cm}X@{}}
\toprule
\textbf{Category} & \textbf{Symbol} & \textbf{Description} \\
\midrule
\multirow{8}{*}{\parbox{1.3cm}{\textbf{Para-meters}}}
& $\numparams{\playeridx}$ & Parameter dimension for agent $\playeridx$ (positive integer) \\
& $\paramspace{\playeridx} \subseteq \mathbb{R}^{\numparams{\playeridx}}$ & Parameter space for agent $\playeridx$ \\
& $\jointparamspace \equiv \paramspace{1} \times \cdots \times \paramspace{\numplayers}$ & Joint parameter space (all agents) \\
& $\gtparams{\playeridx} \in \paramspace{\playeridx}$ & Level-1 parameter vector for agent $\playeridx$ \\
& $\GtParams \equiv \{\gtparams{1}, \gtparams{2}, \ldots, \gtparams{\numplayers}\}$ & Joint level-1 parameter set for all agents \\
& $\gtparams{*}$ & Ground truth parameters (level-1 inference) \\
& $\params{\playeridx} \in \paramspace{\playeridx}$ & Estimate of agent $\playeridx$'s parameters \\
& $\Params = \{\params{1}, \params{2}, \ldots, \params{\numplayers}\}$ & Estimated parameters for all agents \\
\midrule
\multirow{2}{*}{\textbf{Costs}}
& $\stagecostfn{t}{\playeridx}(\xstate{t}{}, \ctrl{t}{}; \gtparams{\playeridx})$ & Stage cost for agent $\playeridx$ at time $t$ \\
& $\costfn{\playeridx}(\StateTraj{}, \CtrlTraj{}; \gtparams{\playeridx})$
& Total cost for agent $\playeridx$ over horizon \\
% \cref{eq:objective}
\midrule
\multirow{4}{*}{\parbox{1.3cm}{\textbf{Con-straints}}}
& $\numequalitylagrange$ & Total number of equality constraints \\
& $\numconstlagrange$ & Total number of inequality constraints \\
& $\equalityconstraint{\playeridx}(\StateTraj{}, \CtrlTraj{}) = 0$ & Equality constraint function for agent $\playeridx$ \\
& $\inequalityconstraint{\playeridx}(\StateTraj{}, \CtrlTraj{}) \geq 0$ & Inequality constraint function for agent $\playeridx$ \\
% & $\equalityConstraintLagrange{\playeridx}$ & Lagrange multipliers for equality constraints \\
% & $\inequalityConstraintLagrange{\playeridx} \geq 0$ & Lagrange multipliers for inequality constraints \\
\midrule
\multirow{5}{*}{\parbox{1.3cm}{\textbf{Nash Game}}}
& $\Gamma(\GtParams, \initstate{}, \horizon)$ & $\numplayers$-agent parameterized Nash game \\
% from \cref{eq:nash-game} \\
& $\Gamma(\GtParams)$ & Simplified game notation (when $\initstate{}, \horizon$ clear from context) \\
& $(\StateTraj{*}, \CtrlTraj{*})$ & Local generalized Nash equilibrium (LGNE) \\
& $\StateTraj{*}$ & LGNE state trajectory \\
& $\CtrlTraj{*}$ & LGNE strategy (equilibrium control trajectory) \\
\bottomrule
\end{tabularx}
\label{tab:notation_costs_params}
\end{table*}

\begin{table*}[t]
\caption{Notation Reference III: Level-2 Formulation}
\small
% \begin{tabular}{@{}p{1.3cm}p{3.2cm}p{9cm}@{}}
\begin{tabularx}{\textwidth}{@{}p{2.5cm}p{4cm}X@{}}
\toprule
\textbf{Category} & \textbf{Symbol} & \textbf{Description} \\
\midrule
\multirow{9}{*}{\parbox{1.3cm}{\textbf{Level-2 Params}}}
& $\gtparams{\playeridx,\playeridx}$ & Agent $\playeridx$'s true objective parameter (known to $\playeridx$) \\
& $\gtparams{\playeridx,-\playeridx} $ & Agent $\playeridx$'s estimates of other agents' parameters ($\{\gtparams{\playeridx,1}\!\!,..., \gtparams{\playeridx,\playeridx-1}\!, \gtparams{\playeridx,\playeridx+1}\!\!,...,\gtparams{\playeridx,\numplayers}\}$)  \\
% $\!\! \equiv \{\gtparams{\playeridx,1}\!\!\!\!\!,..., \gtparams{\playeridx,\playeridx-1}\!\!, \gtparams{\playeridx,\playeridx+1}\!\!\!\!\!\!\!\!,...,\gtparams{\playeridx,\numplayers}\}$ 
% &$\equiv \$ & \\
% $\equiv \{\gtparams{\playeridx,1}, \ldots, \gtparams{\playeridx,\playeridx-1}, \gtparams{\playeridx,\playeridx+1}, \ldots, \gtparams{\playeridx,\numplayers}\}$ 
% & (Agent $\playeridx$'s estimates of other agents' parameters) \\
& $\secondgtparams{\playeridx} \equiv \{\gtparams{\playeridx,\playeridx}, \gtparams{\playeridx,-\playeridx}\} \in \jointparamspace$ & Agent $\playeridx$'s full parameter set in level-2 game \\
& $\secondgtparams{} \equiv \{\secondgtparams{1}, \secondgtparams{2}, \ldots, \secondgtparams{\numplayers}\}$ & Complete set of all agents' level-2 parameters (in $\jointparamspace \times \cdots \times \jointparamspace$) \\
& $\secondgtparams{*}$ & Ground truth level-2 parameters \\
& $\params{\playeridx,\playeridx}$ & Estimate of agent $\playeridx$'s own objective parameter \\
& $\params{\playeridx,-\playeridx}$ & Estimate of agent $\playeridx$'s estimates of others' objectives \\
& $\GtParamsguess{\playeridx} \equiv \{\params{\playeridx,\playeridx}, \params{\playeridx,-\playeridx}\} \in \jointparamspace$ & Third-party observer's estimate of $\secondgtparams{\playeridx}$ \\
& $\GtParamsguess{} \equiv \{\GtParamsguess{1}, \GtParamsguess{2}, \ldots, \GtParamsguess{\numplayers}\}$ & All parameters to infer in level-2 inverse game (in $\jointparamspace \times \cdots \times \jointparamspace$) \\
\midrule
\multirow{6}{*}{\parbox{1.3cm}{\textbf{Hypothe-sized Equilibria}}}
% & $\Gamma(\secondgtparams{\playeridx})$ or
& $\Gamma(\GtParamsguess{\playeridx})$ & Agent $\playeridx$'s hypothesized Nash game \\
& $\hypothesizedStateTraj{\playeridx} \equiv \{\hypStateTraj{\playeridx,\secondidx}\}_{\secondidx=1}^\numplayers$ & Hypothesized LGNE state trajectories in $\Gamma(\secondgtparams{\playeridx})$ \\
& $\hypothesizedCtrlTraj{\playeridx} \equiv \{\hypCtrlTraj{\playeridx,\secondidx}\}_{\secondidx=1}^\numplayers$ & Hypothesized LGNE control strategies in $\Gamma(\secondgtparams{\playeridx})$ \\
& $\hypStateTraj{\playeridx,\secondidx}$ & Player $\secondidx$'s hypothesized LGNE state trajectory in agent $\playeridx$'s game \\
& $\hypCtrlTraj{\playeridx,\secondidx}$ & Player $\secondidx$'s hypothesized LGNE control strategy in agent $\playeridx$'s game \\
& $\hypCtrlTraj{\playeridx,\playeridx} \equiv \{\hypctrl{t}{\playeridx,\playeridx}\}_{t=1}^{\horizon}$ & Agent $\playeridx$'s own controls extracted from $\hypothesizedCtrlTraj{\playeridx}$ \\

& $\hat\zeta^{i*} = (\hat\zeta^{i*}_{\hypothesizedStateTraj{}}, \hat\zeta^{i*}_{\hypothesizedCtrlTraj{}})$ & Player $\playeridx$'s own LGNE state and control trajectories in its own hypothesized game $\Gamma(\GtParamsguess{i})$ \\
% $\todo{explain}\todo{Add $\zeta^i, \zeta^i_x, \zeta^i_u$} \\
\midrule
\multirow{3}{*}{\parbox{1.3cm}{\textbf{Level-2 Trajectories}}}
& $\CtrlTraj{} = \{\hypCtrlTraj{1,1}, \hypCtrlTraj{2,2}, \ldots, \hypCtrlTraj{\numplayers,\numplayers}\}$ & Joint level-2 control trajectory (each agent uses its own hypothesized control) \\
& $\StateTraj{} = \{\hypStateTraj{1,1}, \hypStateTraj{2,2}, \ldots, \hypStateTraj{\numplayers,\numplayers}\}$ & Joint level-2 state trajectory (resulting from joint level-2 controls) \\
& & \\
\midrule
\textbf{Loss} & $\loss(\GtParamsguess{})$ & Level-2 inverse game loss function \\
\bottomrule
\end{tabularx}
\label{tab:notation_level2}
\end{table*}

\begin{table*}[t]
\caption{Notation Reference IV: Linear-Quadratic Games}
\small
% \begin{tabular}{@{}p{1.3cm}p{3.2cm}p{9cm}@{}}
\begin{tabularx}{\textwidth}{@{}p{2.5cm}p{4cm}X@{}}
\toprule
\textbf{Category} & \textbf{Symbol} & \textbf{Description} \\
\midrule
\multirow{3}{*}{\textbf{Costs (LQ)}}
& $Q(\gtparams{\playeridx})$ & Positive semidefinite state cost matrix parameterized by $\gtparams{\playeridx}$ \\
& $R^{\playeridx}$ & Fixed positive definite control cost matrix for agent $\playeridx$ \\
& $\stagecostfn{t}{\playeridx}(\xstate{t}{}, \ctrl{t}{}; \gtparams{\playeridx})$ & Quadratic stage cost \cref{eq:lq-cost} \\
% $= \frac{1}{2}\xstate{t}{}{}^\T Q(\gtparams{\playeridx})\xstate{t}{} + \frac{1}{2}(\ctrl{t}{\playeridx}){}^\T R^{\playeridx} \ctrl{t}{\playeridx}$ (Quadratic stage cost) \\
\midrule
\multirow{2}{*}{\parbox{1.3cm}{\textbf{ Dynamics (LQ)}}}
& $\LinearKKTA^{\playeridx}, \LinearKKTB^{\playeridx}$ & State transition and control input matrices for agent $\playeridx$ \\
& $\xstate{t+1}{\playeridx} = \LinearKKTA^{\playeridx} \xstate{t}{\playeridx} + \LinearKKTB^{\playeridx} \ctrl{t}{\playeridx}$ & Linear dynamics for agent $\playeridx$ \\
\midrule
\multirow{8}{*}{\parbox{1.3cm}{\textbf{Dual and Total Variables (LQ)}}}
& $\dynamicsLagrange{t}{\playeridx}$ & Lagrange multiplier for dynamics constraint at time $t$ \\
& $\dynamicsLagrange{t}{} \equiv \{\dynamicsLagrange{t}{\playeridx}\}_{\playeridx=1}^\numplayers$ & Collection of dynamics multipliers at time $t$ \\
& $\DynamicsLagrange{} \equiv \{\dynamicsLagrange{t}{}\}_{t=1}^{\horizon-1}$ & All dynamics multipliers across time \\
& $\dynamicsInitLagrange{\playeridx}$ & Lagrange multiplier for initial condition $0 = \xstate{1}{\playeridx} - \initstate{\playeridx}$ \\
& $\DynamicsInitLagrange{} \equiv \{\dynamicsInitLagrange{\playeridx}{}\}_{\playeridx=1}^\numplayers$ & Collection of all initial condition multipliers \\
% \midrule
% \multirow{3}{*}{\parbox{1.3cm}{\textbf{LQ Variables}}}
& $\LQtotalvariable \equiv [\StateTraj{}, \CtrlTraj{}, \DynamicsLagrange{}, \DynamicsInitLagrange{}]$ & All primal and dual variables for LQ game \\
& $n_{\bar{z}}$ & Total dimension of $\LQtotalvariable$ ($\horizon(\numstates{} + \numctrls{}) + (\horizon-1)\numstates{} + \numstates{}$) \\
& $\secondLQtotalvariable^{\playeridx} \in \mathbb{R}^{n_{\bar{z}}}$ & Primal/dual variables for agent $\playeridx$'s hypothesized LQ game \\
% \midrule
% \multirow{2}{*}{\parbox{1.3cm}{\textbf{Lagrangian (LQ)}}}
\multirow{11}{*}{\parbox{1.3cm}{\textbf{KKT System Variables (LQ)}}}
& $\LQlagrangian{\playeridx}(\StateTraj{}, \CtrlTraj{}, \DynamicsInitLagrange{}, \DynamicsLagrange{}; \gtparams{\playeridx})$ 
% $\equiv \costfn{\playeridx}(\StateTraj{}, \CtrlTraj{}; \gtparams{\playeridx}) + \dynamicsInitLagrange{\playeridx}(\xstate{1}{\playeridx} - \initstate{\playeridx})$ \\
% & & $+ \sum_{t=1}^{\horizon-1} \dynamicsLagrange{t}{\playeridx}(\xstate{t+1}{\playeridx} - \LinearKKTA^{\playeridx} \xstate{t}{\playeridx} - \LinearKKTB^{\playeridx} \ctrl{t}{\playeridx})$ (
& Lagrangian for agent $\playeridx$ in LQ game \\
\midrule
& $M(\GtParams)$ & KKT matrix encoding linear system for LQ game equilibrium \\
& $S$ & Selection matrix for initial conditions in KKT system \\
& $\bar{Q}(\GtParams)$ & Block-repeated state cost matrix in $M(\GtParams)$ (size $\horizon \numstates{} \times \horizon \numstates{}$) \\
& $\bar{R}$ & Block-repeated control cost matrix in $M(\GtParams)$ (size $\horizon \numctrls{} \times \horizon \numctrls{}$) \\
& $\LinearKKTBlkDiagA \equiv \text{diag}(\LinearKKTA^1, \LinearKKTA^2, \ldots, \LinearKKTA^\numplayers)$ & Block diagonal dynamics matrix \\
& $\LinearKKTBlkDiagB \equiv \text{diag}(\LinearKKTB^1, \LinearKKTB^2, \ldots, \LinearKKTB^\numplayers)$ & Block diagonal control matrix \\
& $\LinearKKTBarA$ & Block matrix with $-\LinearKKTA^{\playeridx}$ and $I$ for dynamics constraints (size $(\horizon-1)\numstates{} \times \horizon\numstates{}$) \\
& $\LinearKKTBarB$ & Block matrix with $-\LinearKKTB^{\playeridx}$ for dynamics constraints (size $(\horizon-1)\numstates{} \times \horizon\numctrls{}$) \\
& $\LinearKKTBarC$ & Selection matrix for initial conditions (size $\numstates{} \times \horizon\numstates{}$) \\
& $\LinearKKTBigBlkDiagA, \LinearKKTBigBlkDiagB, \LinearKKTBigBlkDiagC$ & Transposed/structured versions for dual blocks in $M(\GtParams)$ \\
\midrule
\multirow{2}{*}{\parbox{1.3cm}{\textbf{Solution (LQ)}}}
& $M(\GtParams)\LQtotalvariable + S\initstate{} = 0$ & Compact KKT system for LQ game \\
& $\LQcontrolfunction(\GtParamsguess{\playeridx})$ & Agent $\playeridx$'s LGNE control as function of $\GtParamsguess{\playeridx}$ (subvector of $\secondLQtotalvariable^{\playeridx}$) \\
\bottomrule
\end{tabularx}
\label{tab:notation_lq}
\end{table*}

\begin{table*}[t]
\caption{Notation Reference V: Homogeneous Parameters and Error Bounds}
\small
% \begin{tabular}{@{}p{1.3cm}p{3.2cm}p{9cm}@{}}
\begin{tabularx}{\textwidth}{@{}p{2.5cm}p{4cm}X@{}}
\toprule
\textbf{Category} & \textbf{Symbol} & \textbf{Description} \\
\midrule
\multirow{7}{*}{\parbox{1.3cm}{\textbf{Homo-geneous Level-2 Params}}}
& $\LevelOneGtParamsguess$ & Homogeneous level-2 parameter (level-1 parameters in level-2 parameter space) \\
& $\LevelOneGtParamsguess^{\playeridx} = \LevelOneGtParamsguess^{\secondidx} = \GtParams$ & All agents have identical parameter estimates (level-1 assumption) \\
& $\LevelOneGtParamsguess^*$ & Optimal homogeneous level-2 parameter minimizing $\loss(\LevelOneGtParamsguess)$ \\
& $\constructedlevelonesol^{\playeridx} \equiv \gtparams{*}$ & Homogeneous level-2 params constructed from level-1 ground truth \\
& $\constructedlevelonesol \equiv \{\constructedlevelonesol^{\playeridx}\}_{\playeridx=1}^\numplayers$ & Collection of all constructed homogeneous parameters \\
& $\check{M} \equiv \left(\sum_{\playeridx=1}^\numplayers \frac{1}{\numplayers}M^{-1}(\secondgtparams{\playeridx*})\right)^{-1}$ & Average inverse LQ KKT condition matrix \\
% & $\sigma_{\min}(\cdot), \sigma_{\max}(\cdot)$ & Smallest and largest singular values of a matrix \\
& $E^{\playeridx}$ & Selection matrix extracting agent $\playeridx$'s control from $\LQtotalvariable$ \\
\bottomrule
\end{tabularx}
\label{tab:notation_homogeneous}
\end{table*}

\begin{table*}[t]
\caption{Notation Reference VI: MCP Formulation}
\small
% \begin{tabular}{@{}p{1.3cm}p{3.2cm}p{9cm}@{}}
\begin{tabularx}{\textwidth}{@{}p{2.5cm}p{4cm}X@{}}
\toprule
\textbf{Category} & \textbf{Symbol} & \textbf{Description} \\
\midrule
% \multirow{2}{*}{\parbox{1.3cm}{\textbf{MCP Constr. Counts}}}
% & $\numequalitylagrange$ & Total number of equality constraint multipliers \\
% & $\numconstlagrange$ & Total number of inequality constraint multipliers \\
% \midrule
% \multirow{4}{*}{\parbox{1.3cm}{\textbf{MCP Mult. (Agent $\playeridx$'s Game)}}}
\multirow{9}{*}{\parbox{1.3cm}{\textbf{Dual and Total Variables (MCP)}}}
& $\hypdynamicsLagrange{t}{\playeridx,\secondidx} \in \mathbb{R}^{\numstates{}}$ & Dynamics multiplier at time $t$ for agent $\secondidx$ in $\playeridx$'s hypothesized game \\
& $\hypDynamicsLagrange{\playeridx,\secondidx}$ & Dynamics multipliers for agent $\secondidx$ in $\playeridx$'s hypothesized game over horizon $\horizon$ \\
& $\hypDynamicsLagrange{\playeridx}$ & All dynamics multipliers for agent $\playeridx$'s hypothesized game \\
& $\hypdynamicsInitLagrange{\playeridx,\secondidx} \in \mathbb{R}^{\numstates{}}$ & Initial state multiplier for agent $\secondidx$ in $\playeridx$'s hypothesized game \\
& $\hypDynamicsInitLagrange{\playeridx} \equiv \{\hypdynamicsInitLagrange{\playeridx,\secondidx}\}_{\secondidx=1}^\numplayers$ & All initial state multipliers for agent $\playeridx$'s hypothesized game \\
& $\hypequalityConstraintLagrange{\playeridx,\secondidx} \in \mathbb{R}^{\numequalitylagrange}$ & Equality constraint multipliers for agent $\secondidx$ in $\playeridx$'s hypothesized game \\
& $\hypEqualityConstraintLagrange{\playeridx} \equiv \{\hypequalityConstraintLagrange{\playeridx,\secondidx}\}_{\secondidx=1}^\numplayers$ & All equality multipliers for agent $\playeridx$'s game \\
& $\hypinequalityConstraintLagrange{\playeridx,\secondidx} \in \mathbb{R}^{\numconstlagrange}$ & Inequality constraint multipliers for agent $\secondidx$ in $\playeridx$'s hypothesized game \\
% \midrule
% \multirow{3}{*}{\parbox{1.3cm}{\textbf{MCP Aggr. Mult.}}}
& $\hypInequalityConstraintLagrange{\playeridx} \equiv \{\hypinequalityConstraintLagrange{\playeridx,\secondidx}\}_{\secondidx=1}^\numplayers$ & All inequality multipliers for agent $\playeridx$'s game \\
% \midrule
% \multirow{2}{*}{\parbox{1.3cm}{\textbf{MCP Total Vars.}}}
& $\totalvariable_{\playeridx} \equiv [\hypothesizedStateTraj{\playeridx}, \hypothesizedCtrlTraj{\playeridx}, \hypDynamicsLagrange{\playeridx}, \hypDynamicsInitLagrange{\playeridx}, \hypEqualityConstraintLagrange{\playeridx}, \hypInequalityConstraintLagrange{\playeridx}]$ & All primal/dual variables for agent $\playeridx$'s hypothesized game \\
% \midrule
% \textbf{MCP Lag.} 
\midrule
\multirow{5}{*}{\parbox{1.3cm}{\textbf{Problem (MCP)}}} & $\lagrangian{\playeridx,\secondidx}(\totalvariable_{\playeridx}; \params{\playeridx,\secondidx})$ & Lagrangian for agent $\secondidx$ in agent $\playeridx$'s hypothesized game $\Gamma(\GtParamsguess{\playeridx})$ \\
% \midrule
% \multirow{5}{*}{\parbox{1.3cm}{\textbf{MCP Problem}}}
& $F_{\text{eq}}(\totalvariable_{\playeridx}; \GtParamsguess{\playeridx})$ & MCP equality constraint and stationarity vector \\
& $F_{\text{ineq}}(\totalvariable_{\playeridx}; \GtParamsguess{\playeridx})$ & MCP inequality constraint vector \\
% & $0 = F_{\text{eq}}(\totalvariable_{\playeridx}; \GtParamsguess{\playeridx})$ & MCP equality conditions \\
& $0 \leq \hypInequalityConstraintLagrange{\playeridx} \perp F_{\text{ineq}}(\totalvariable_{\playeridx}; \GtParamsguess{\playeridx}) \geq 0$ & MCP complementarity conditions \\
& $\totalvariable^*_{\playeridx}$ & Solution to MCP (LGNE of $\Gamma(\GtParamsguess{\playeridx})$) \\
\midrule
\multirow{3}{*}{\parbox{1.3cm}{\textbf{Gradients (MCP)}}}
& $\nabla \totalvariable^*_{\playeridx}(\GtParamsguess{\playeridx})$ & Gradient of LGNE solution w.r.t. $\GtParamsguess{\playeridx}$ (via implicit function theorem) \\
& $\thetavec$ & Vectorized form of $\GtParamsguess{}$ for gradient computation \\
& $\xstate{t}{(\playeridx,\playeridx)*}$ & Agent $\playeridx$'s LGNE state at time $t$ in $\Gamma(\GtParamsguess{\playeridx})$ \\
\bottomrule
\end{tabularx}
\label{tab:notation_mcp}
\end{table*}

\begin{table*}[t]
\caption{Notation Reference VII: Algorithm and Lane Change Example}
\small
% \begin{tabular}{@{}p{1.3cm}p{3.2cm}p{9cm}@{}}
\begin{tabularx}{\textwidth}{@{}p{2.5cm}p{4cm}X@{}}
\toprule
\textbf{Category} & \textbf{Symbol} & \textbf{Description} \\
\midrule
\multirow{6}{*}{\textbf{Algorithm}} 
% \todo{Which of these are not needed anymore?}
& $K$ & Maximum number of iterations for Algorithm 1 \\
& $\epsilon > 0$ & Convergence threshold for loss function \\
& $\alpha > 0$ & Step size (learning rate) for gradient descent \\
& $\GtParamsguess{}{'}$ & Updated parameter estimate after gradient descent step \\
& $\nabla \loss(\GtParamsguess{})$ &
Gradient of loss function with respect to $\GtParamsguess{}$ \\
& $\frac{\partial \loss}{\partial \thetavec_{\secondidx}}$ & Chain rule gradient with respect to parameter $\thetavec$, defined in \cref{eq:loss gradient} \\
% $= \sum_{\playeridx=1}^\numplayers \sum_{t=1}^\horizon \frac{\partial \loss}{\partial \ObservationFn(\xstate{t}{(\playeridx,\playeridx)*})} \cdot \frac{\partial \ObservationFn(\xstate{t}{(\playeridx,\playeridx)*})}{\partial \xstate{t}{(\playeridx,\playeridx)*}} \cdot \frac{\partial \xstate{t}{(\playeridx,\playeridx)*}}{\partial \thetavec_{\secondidx}}$ (Chain rule gradient) \\
\midrule
\multirow{12}{*}{\parbox{1.3cm}{\textbf{Lane Change Example}}}
& $l_w \in \mathbb{R}$ & Lane width (meters) \\
& $v_d^{\playeridx} \in \mathbb{R}^2$ & Desired velocity for agent $\playeridx$ \\
& $\Delta t$ & Discrete time sampling interval (seconds) \\
& $p_t^{\playeridx} = [p_{\text{lat},t}^{\playeridx} ~ p_{\text{lon},t}^{\playeridx}]^\T \in \mathbb{R}^2$ & Agent $\playeridx$'s planar position at time $t$ (lateral and longitudinal) \\
& $v_t^{\playeridx} = [v_{\text{lat},t}^{\playeridx} ~ v_{\text{lon},t}^{\playeridx}]^\T \in \mathbb{R}^2$ & Agent $\playeridx$'s planar velocity at time $t$ \\
& $\ctrl{t}{\playeridx} = [a_{\text{lat},t}^{\playeridx} ~ a_{\text{lon},t}^{\playeridx}]^\T$ & Agent $\playeridx$'s commanded planar acceleration/control at time $t$ \\
& $\xstate{t+1}{\playeridx} = [p_{t+1}^{\playeridx} ~ v_{t+1}^{\playeridx}]^\T$ & Double-integrator dynamics \cref{eq:exp-planar-di-dynamics} \\
& $w_{\secondidx}^{\playeridx} > 0$ & Weight for agent $\playeridx$'s cost component $\secondidx$ \\
& $\stagecostfn{t}{\playeridx}(\xstate{t}{}, \ctrl{t}{}; \gtparams{\playeridx,\playeridx})$ & Lane change stage cost defined using \cref{eq:lc-objective,eq:lc-weight-1,eq:lc-weight-2,eq:lc-weight-3} \\
% $\equiv w_1^{\playeridx}(p_{\text{lat},t}^{\playeridx} - \gtparams{\playeridx,\playeridx})^2 + w_2^{\playeridx} \|v_t^{\playeridx} - v_d\|_2^2 + w_3^{\playeridx} \|\ctrl{t}{\playeridx}\|_2^2$ 
& $\delta \in \mathbb{R}$ & Safety buffer distance for collision avoidance (meters) \\
% & $\|p_t^1 - p_t^2\|_2^2 \geq \delta^2, \forall t \in [\horizon]$ & Collision avoidance constraint \\
& $\secondgtparams{*}$ & Ground truth level-2 parameters for lane-change scenario \\
% , defined in \cref{eq:lane-change-gt-params} \\
% $\equiv \{\{\gtparams{(1,1)*}, \gtparams{(1,2)*}\}, \{\gtparams{(2,1)*}, \gtparams{(2,2)*}\}\}$ (Ground truth parameters for two-agent scenario) \\
& $\gtparams{(1,1)*} = \gtparams{(2,2)*} = l_w/2$ & Both agents' true desired lane offsets (midpoint of top lane) \\
\bottomrule
\end{tabularx}
\label{tab:notation_algorithm_example}
\end{table*}